\pdfoutput=1
\documentclass[journal]{IEEEtran}
\IEEEoverridecommandlockouts
\usepackage{xcolor}
\usepackage{cite}
\usepackage{amsmath,amssymb,amsfonts}
\usepackage{graphicx}
\usepackage{epstopdf}
\usepackage{textcomp}
\usepackage{xcolor}
\usepackage{amsmath,amssymb,epsfig,subfigure,amsthm,mathtools}
\usepackage{cite}
\usepackage{color}
\usepackage{url}
\usepackage{color,soul}
\usepackage[nocomma]{optidef}
\usepackage{stfloats}
\usepackage{mathrsfs}

\DeclareGraphicsExtensions{.eps}

\newcommand\scalemath[2]{\scalebox{#1}{\mbox{\ensuremath{\displaystyle #2}}}}

\usepackage[linesnumbered, ruled]{algorithm2e}
\RestyleAlgo{ruled}
\usepackage{algorithmic}

\usepackage{mathtools, nccmath}

\newtheorem{lemma}{Lemma}

\newtheorem{proposition}{Proposition}

\SetKwInput{KwInput}{Input} 
\SetKwInput{KwOutput}{Output}

\newcommand{\be}{\begin{equation}}
\newcommand{\ee}{\end{equation}}
\newcommand{\bea}{\begin{eqnarray}}
\newcommand{\eea}{\end{eqnarray}}
\newcommand{\bdp}{\begin{displaymath}}
\newcommand{\edp}{\end{displaymath}}
\usepackage{amsmath}

\ifCLASSINFOpdf
\else
\fi
\begin{document}
%
\title{Low Complexity Algorithms for Mission Completion Time Minimization in UAV-Based ISAC Systems}
%
%
%

\author{Mateen Ashraf, Anna Gaydamaka, Bo Tan, Dmitri Moltchanov, Yevgeni Koucheryavy
\thanks{The authors are with Tampere University.}
}

%
%

\markboth{Journal of \LaTeX\ Class Files, October~2022}%
{Shell \MakeLowercase{\textit{et al.}}: Bare Demo of IEEEtran.cls for IEEE Journals}
%



\maketitle

\vspace{-2 cm}
\begin{abstract}
The inherent support of sixth-generation (6G) systems enabling integrated sensing and communications (ISAC) paradigm greatly enhances the application area of intelligent transportation systems (ITS). One of the mission-critical applications enabled by these systems is disaster management, where ISAC functionality may not only provide localization but also provide users with supplementary information such as escape routes, time to rescue, etc. In this paper, by considering a large area with several locations of interest, we formulate and solve the optimization problem of delivering task parameters of the ISAC system by optimizing the UAV speed and the order of visits to the locations of interest such that the mission time is minimized. The formulated problem is a mixed integer non-linear program which is quite challenging to solve. To reduce the complexity of the solution algorithms, we propose two circular trajectory designs. The first algorithm finds the optimal UAV velocity and radius of the circular trajectories. The second algorithm finds the optimal connecting points for joining the individual circular trajectories. Our numerical results reveal that, with practical simulation parameters, the first algorithm provides a time saving of at least $20\%$, while the second algorithm cuts down the total completion time by at least $7$ times.
\end{abstract}

\begin{IEEEkeywords}
Integrated sensing and communication (ISAC), trajectory design, unmanned aerial vehicle (UAV).
\end{IEEEkeywords}

%
\IEEEpeerreviewmaketitle
\vspace{-4mm}
\section{Introduction}

\IEEEPARstart{I}{ntelligent} transportation systems (ITS) are evolving rapidly over the last decade \cite{sumalee2018smarter}. The use of connected aerial vehicles opens a door to extend the current transportation system in vertical which is an extra  dimension for new applications. The autonomous individual or fleets of unmanned aerial vehicles (UAVs) have been firmly believed to be promising for large-scale landscape inspections, detection and localization of objects, providing enhancement or provision of connectivity services and swift logistics for the distribution on occasions when access with conventional vehicles is difficult \cite{jung2020conceptual}. Large UAVs like electric vertical take-off and landing (eVTOL) vehicles are even capable of personnel and heavy-load goods logistic. One of applications of verticalized future ITS is disaster management. These applications are classified as mission-critical ones \cite{ahmad2019towards} and are useful in earthquakes, flooding, etc. In such applications the objective is not only to navigate rescue teams but also to provide victims with some additional routing information for escape. To this aim, the system should perform both communications and sensing functions simultaneously.


The future sixth-generation (6G) cellular systems promise to deliver integrated sensing and communications (ISAC) \cite{liu2022integrated,cui2021integrating}. The sensing in the mobile network is a collection of functions, such as detection, and estimation of distance motion status of the cooperative or non-cooperative objects. We categorize the ISAC into cooperative and non-cooperative according to whether sensing functions are performed with the help of direct transmission from active radio signal emitters or with the help of scattered signals from passive objects. The cooperative ISAC is often considered the complementary means for the areas where the conventional positioning methods are denied or needed to be enhanced, for example, the GPS in urban canyons \cite{GPS_urban_canyon} and localizing of the 5G-equipped aircraft \cite{BoSun_NewSense}. The non-cooperative ISAC is analogous to performing the active and passive radar while simultaneously carrying communication on unified radio waveforms.  
The integration of ISAC functionalities in aerial vehicles, such as UAVs, can greatly enhance the efficiency of operations in the future verticalized ITS, for example, object detection and trajectory prediction for collision avoidance \cite{Iuliia_TVT}, unmanned aircraft system traffic management (UTM) \cite{ISAC_UTM}, provision of the connectivity and positioning information for the ground transportation participants from the space \cite{ISAC_UAVNet_ITS}.

In this paper, we optimize the use of UAVs for emergency and rescue operations. As opposed to the large set of other studies, we assume that UAVs perform two cooperative ISAC functions: (i) localization of victims and (ii) providing them with additional information such as routes to escape, time to rescue, etc. These functions are performed simultaneously with the help of a UAV via ISAC scanning the disaster areas. By considering a large area with several locations of interest, we formulate and solve the optimization problem where the objective is to minimize the ISAC mission completion time by optimizing the UAV trajectory. 

The main contributions of our work are given as follows:
\begin{itemize}
    \item Mathematical formulation of the positioning and communication problem into a single optimization problem where the objective is to minimize the total mission completion time.
    \item Low complexity algorithms are proposed for solving the formulated problem under different topographical conditions. Specifically, two low complexity algorithms are proposed where the first algorithm solves the optimization problem by providing a single circular trajectory for UAV while the second algorithm solves the optimization problem by providing multiple circular trajectories for individualized smaller scanned regions and subsequently provides the trajectory for interlinking those individual circular trajectories.
    \item It is observed that the proposed algorithms significantly reduce the mission completion time. Furthermore, the saving in mission completion time is either independent of the data threshold or linearly increases with the data threshold depending on the maximum velocity of the UAV. Moreover, in our proposed algorithms, the savings in mission completion time increase with the increase in carrier frequency.
\end{itemize}

The rest of the paper is organized as follows. We review the related work in Section II. The system model and problem formulation are presented in Section III. The proposed algorithms are developed in Section IV. Numerical results are provided in discussed in Section V. Finally, conclusions are drawn in Section VI.

%
%
%
%
 \vspace{-4mm}
\section{Related Work}
The UAV based wireless communication systems have attracted significant attention from the research community. Especially, the performance improvements brought by UAVs in relaying, offloading, and data collection systems have been documented in \cite{zeng1}, \cite{zeng2} and \cite{zeng3}, respectively, for finite number of ground serving points/users. In the recent years, there has been a rapid rise in the implementation of UAVs in search and rescue (SAR) operations \cite{goodrich2009towards}. With the potential for quick monitoring of large areas, UAVs help detect lost people with reduced  cost, time, and risks \cite{WANKMULLER2021102567}. 
The study in \cite{9321707} introduced SARDO, a drone-based solution to localize missing people through mobile phones by using pseudo-trilateration and machine learning. Rescue operations can also occur underwater. For instance, in \cite{MSAR}, the authors proposed the use of UAVs for facilitating localization of multiple targets in the sea zone for SAR operations. The paper \cite{flood} further suggested a deep learning-aided model for detecting people affected by floods.

ISAC systems create an environment where sensing and communication functions perform mutual assistance \cite{cui2021integrating, liu2022integrated}. ISAC solutions achieve higher positioning accuracy, improve wireless communications quality of service and open a brand new range of services \cite{rong20216g}. ISAC is expected to come in handy in many fields such as high-accuracy localization and tracking, simultaneous imaging, mapping, localization, augmented human sense, gesture and activity recognition \cite{9376324}. The work \cite{meng2022uav} mentioned two UAV assisted ISAC deployments: sensing-assisted UAV communication and communication-assisted UAV sensing. The former considers a UAV-to-ground vehicle communication scenario, where vehicle properties (location, velocity) can be extracted from the reflected ISAC signals for beam tracking and alignment. The latter tries to overcome the limited computational ability of UAVs by offloading some computationally-intensive sensing tasks to the central UAV. In \cite{jing2022isac}, authors developed an ISAC framework for the UAV trajectory design with single UAV, UE and a single target. Specifically, the formulated task aims to optimize both the downlink communication rate and localization accuracy.  

Despite the applications and obvious advantages, UAV implementation has a core bottleneck - trajectory planning. Techniques for path planning are usually computationally demanding and infeasible in field experiments.  
In \cite{9641740}, the linear programming formulation was used to model a load-balanced automatic path planning in a heterogeneous UAV swarm. Then, an adaptive clustering-based algorithm was implemented for scanning all the regions in minimum time. 

Without focusing on a specific use case, the study in \cite{9334996} suggested the Improved Bat Algorithm (IBA), which combines the Artificial Bee Colony Algorithm and Bat Algorithm. The IBA takes into account the obstacles while planning the path and outperforms the conventional algorithms in terms of convergence and time complexity. However, the UAV flight path-planning problem in \cite{9334996} was considered only for a static environment.

Besides heuristics methods, machine learning is also adopted for finding the optimal path. Authors in \cite{9121767} investigated the distributed round trip path planning and collision avoidance problems by formulating them as combinatorial/convex optimization and reinforcement learning (RL) problems, respectively. 

The work \cite{9190025} used the deep RL to resolve a mixed-integer non-linear optimization problem which is formulated for UAV path planning and caching in a content delivery system with constraints on UAV trajectory, radio resources, and caching replacement. The solution showed robustness in convergence and energy consumption. Though the solution was claimed to be appliable for drone swarm, only the single UAV result was shown.  

The study in \cite{9316741} also advocated deep RL with biologically inspired algorithms for fast convergence in multi-UAV path planning problems. However, the performance of the solution in the large-scale UAV group needs to be further tested. 

None of the past approaches for UAV trajectory planning provides guaranteed solutions for mission-critical applications such as SAR operations. On top of this, this task has not been addressed in the context of ISAC systems, where the UAVs are not only searching but also supplementing the detected units with additional routing information via the communications channel. These requirements further complicate the problem.
  
\vspace{-2mm}
\section{System model and problem formulation}
This section is divided into four parts. Section III-A provides a detailed discussion of our assumptions about the localization approach used in this paper. Section III-B discusses the assumptions related to the scanning area of interest and the UAV trajectory. Section III-C addresses the signal propagation model. Optimization problem formulation for minimizing the mission completion time is discussed in Section III-D. Moreover, the main system parameters used in this paper are provided in Table I.

\vspace{-0.5 cm}
\subsection{Time of Arrival Based Localization Approach}
We consider the time of arrival (ToA) based localization scheme. Specifically, the UAV transmits beacon sequences which are used to obtain the time differences between the transmitted times and the received times of those beacon sequences. The information about the transmit times, along with the absolute positions from which the beacon signals were transmitted, is encoded into a broadcast message which is transmitted by the UAV. Based on the successful decoding of the broadcast message and the propagation times of the beacon sequences, any receiver node within the scanning area can obtain the absolute distances from the UAV. With the sufficient number of distance calculations based on the beacon sequences, the receiver is able to find its absolute location in three-dimensional (3D) space.

\begin{table}
\begin{center}
\caption{Main system parameters.}
\vspace{-3mm}
\begin{tabular}{ |l| l| }
\hline
System parameter & Notation \\
\hline
\hline
 transmission carrier frequency & $f_c$ \\
 \hline
 transmission bandwidth & $\Delta f$  \\
 \hline
 transmission power & $P$ \\
 \hline
 noise power at the receiver & $\sigma^2$ \\
 \hline
 normalized signal to noise ratio at receiver & $\gamma$ \\
 \hline
 UAV's directional antenna gain in direction $(\phi, \varphi)$ & $G(\phi,\varphi)$ \\
 \hline
 radius of the smallest circle covering the whole area & $R$ \\
 \hline
 center of the scanning area & $\mathbf{c}$ \\
 \hline
 radius of the UAV trajectory & $r_U$ \\
 \hline
 height of the UAV from ground & $H$ \\
 \hline
 position of the UAV at $t$-th time instant & $\mathbf{u}(t)$ \\
 \hline
 total completion time & $T$ \\
 \hline
 number of connected intervals for position $\mathbf{x}$ & $N(\mathbf{x})$ \\
 \hline
 starting time of $n$-th connection interval for position $\mathbf{x}$ & $t_s^n(\mathbf{x})$ \\
 \hline
 ending time of $n$-th connection interval for position $\mathbf{x}$ & $t_e^n(\mathbf{x})$ \\
 \hline
 data transferred to position $\mathbf{x}$ over time $T$ & $\hat{R}(\mathbf{x})$ \\
 \hline
\end{tabular}
\end{center}
\vspace{-4mm}
\end{table}

In the considered ToA-based localization, it is necessary for the UAV to remain mobile. This is explained by the following example for one dimensional (1D) case. Consider that we have a transmitter with a known location and a receiver node with an unknown location. Further, assume that the receiver node wants to know its location. The transmitter broadcasts a signal containing the beacon sequence along with its transmit time and the transmitter's location from which the beacon sequence was transmitted. With the help of the correct reception of the broadcast signal, the receiver knows its absolute distance from the transmitter. However, this information is insufficient to deduce whether it is on the right or left side of the transmitter. To precisely localize the receiver in 1D space in a zero noise scenario, the transmitter must send at least two beacon sequences to remove the ambiguity about the receiver node's left/right direction possibilities with respect to the transmitter. However, if both the beacon sequences are transmitted from the same location, then the receiver will not be able to resolve the left/right direction ambiguity. On the other hand, if the transmitter changes its location, the receiver can resolve the left/right ambiguity. Although this toy example considers 1D localization under zero noise, the observation is also valid for 2D and 3D localization problems under noisy conditions. Hence, in the proposed ToA-based localization, it is necessary for the transmitter to remain mobile.

Thus, the successful localization at any particular location in the coverage area depends on successfully detecting the mobile UAV's broadcast message. Essentially, the localization task is therefore converted into a multicast/broadcast task. However, we emphasize that while the localization task is converted into a multicast/broadcast task, the approach to accomplish it is not the same as used for conventional multicast/broadcast scenarios. This is explained with the help of the following example. Assume that the task is to broadcast a message to all users inside a circular coverage area. While in a conventional broadcast/multicast scenario placing the transmitter at the center for the whole transmission time in a conventional broadcast/multicast scenario may accomplish the task, it is not applicable to the scenario considered in this paper due to the requirement of constant changes in the location of the transmitter.

\vspace{-4mm}
\subsection{UAV Trajectory and Scanning Area}
Owing to the above-mentioned reasons, we assume that the UAV is constantly in motion. Moreover, we consider a circular scanning area where the positioning and communication services must be provided. The radius of the coverage area is denoted by $R$ while its center is denoted by $\mathbf{c}$. For scenarios with a non-circular scanning area, the scanning area can be approximated by the smallest circle that encompasses the whole scanning area. The scanning area consists of open outdoor and indoor regions due to the building blocks. During the whole mission time, the UAV flies at an altitude of $H$ meters. The exact value of $H$ is chosen to avoid collisions with the buildings. Specifically, if $H_b$ is the height of the tallest building in the coverage area, we chose $H > H_b$\footnote{The height of the UAV can be chosen to be $H = H_b+H_m$ where $H_m$ is a suitable margin.}. On the other hand, the $x$, $y$ coordinates of the UAV at any given time are chosen according to the trajectory. An illustration of the considered scenario is shown in Fig. \ref{fig:scenarion_main}.
\begin{figure}[h]
\centering
\includegraphics[width=\columnwidth]{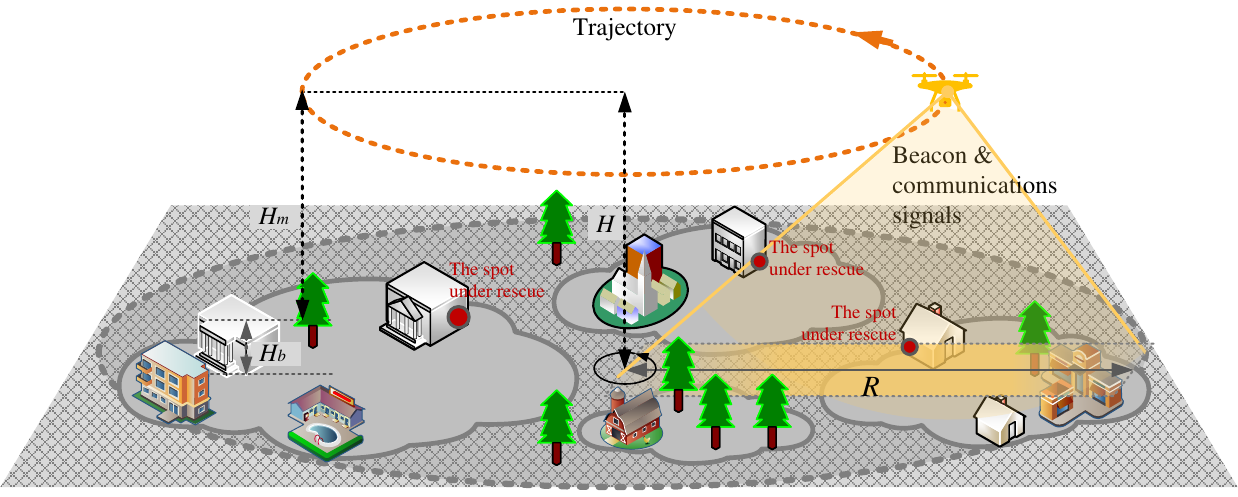}
\caption{An example of the coverage region topography. The red dots represent the possible locations of the to-be-rescued.}
\label{fig:scenarion_main}
\vspace{-6mm}
\end{figure}

\subsection{Signal Model Assumptions}
We assume that the received power at any location within the coverage area is dependent on two factors: (i) antenna gain and (ii) propagation loss. For antenna gain, we assume that the UAV is equipped with a directional antenna where the gain is given by \cite{belanis}
\vspace{-2mm}
\begin{equation}
    G(\phi,\varphi) = \Bigg\{\begin{array}{lr}
         \frac{G_0}{\Phi_a \phi_e} & -\frac{\pi}{2} \leq \phi_e \leq \frac{\pi}{2}, -\Phi_a \leq \varphi \leq \Phi_a \\
         g_0\simeq 0, & \text{otherwise}, 
    \end{array}
\end{equation}
where $G_0=\left(\frac{\sqrt{7500}\pi}{180}\right)^2$, and $\varphi, \phi_e$ denote the azimuth and elevation angle, respectively. With this antenna gain model, the ground area over which the received power is non-zero is represented by a rectangular strip. 

For propagation loss, we assume power loss during propagation is caused by two factors: (i) distance-based path loss and (ii) building/wall penetration loss. Assuming that the horizontal distance between the UAV and the receiver is $d$, the distance-based path loss is given as \cite{zeng3}
\vspace{-2mm}
\begin{equation}
    l_{db}(d, f_c)= \frac{c^2}{\left(4\pi f_c\sqrt{(d^2+H^2)}\right)^2},
\end{equation}
where $f_c$ (GHz) denotes the carrier frequency and $c$ represents the speed of light. The power loss caused by penetrations is largely dependent on the carrier frequency and the building material/thickness of the walls. In this paper, we assume that the penetration loss is given by \cite{ETSI} 
\vspace{-1mm}
\begin{equation}
    l_{p}(f_c) = 10^{.5+.4f_c}.
\end{equation}
\vspace{-2mm}
The total loss is thus given as 
\begin{equation}
    l_T(d,f_c)= l_{db}(d,f_c)\times l_p(f_c).
\end{equation}
Since we have assumed a mobile UAV, the horizontal distance, elevation angle and azimuth angle for the UAV with respect to any fixed location on the ground continuously change. Assuming that at time $t$, the UAV position is given by $\mathbf{u}(t)$, a ground location $\mathbf{x}$ has a horizontal distance $d_{\mathbf{x}}(t)=\|\mathbf{x}-\mathbf{u}(t)\|$ from the UAV\footnote{Note since $\mathbf{u}(t)$ is dependent on time, we write the distance as a function of time.}, and the corresponding azimuth, elevation angles are given by $\phi_e(\mathbf{x},t), \varphi(\mathbf{x},t)$, respectively. Then, the received signal power at location $\mathbf{x}$ can be written as (5) (shown at the top of this page),
\begin{figure*}
\begin{equation}
    P_r(\mathbf{x},t)=\Bigg\{\begin{array}{lr}\frac{PG(\phi_e(t),\varphi(t))}{l_T(t)}= \frac{Pl_p(f_c)G_0c^2}{\Phi_a|\phi_e(\mathbf{x},t)| \left(4\pi f_c\sqrt{(d_{\mathbf{x}}^2(t)+H^2)}\right)^2}, -\frac{\pi}{2} \leq \phi_e(\mathbf{x},t) \leq \frac{\pi}{2}, -\Phi_a \leq \varphi(\mathbf{x},t) \leq \Phi_a, & \\
    0, ~~~~~~ \text{otherwise}, & 
    \end{array}
\end{equation}
\begin{equation}
    R(\mathbf{x},t) = \Bigg\{\begin{array}{lr} \log_2\left(1+\frac{\gamma l_p(f_c)G_0c^2}{\Phi_a|\phi_e (\mathbf{x},t)| \left(4\pi f_c\sqrt{(d_{\mathbf{x}}^2(t)+H^2)}\right)^2}\right),-\frac{\pi}{2} \leq \phi_e(\mathbf{x},t) \leq \frac{\pi}{2}, -\Phi_a \leq \varphi(\mathbf{x},t) \leq \Phi_a, & \\
    0, ~~~ \text{otherwise}, &
    \end{array}
\end{equation}
\hrule
\end{figure*}
where $P$ is the transmit power of the UAV. Using the Shannon capacity formula, the corresponding data rate achievable at time $t$ can be written as (6) (shown at the top of this page), where $\gamma = \frac{P}{\sigma^2}$, and $\sigma^2$ is the additive white Gaussian noise variance at the receiver.

Next, assume that a particular location $\mathbf{x}$ starts and stops receiving non-zero power for the $n$-th time at $t_s^{n}(\mathbf{x})$ and $t_e^{n}(\mathbf{x})$, respectively. Then, the minimum total data rate that can be transferred to this location can be written as
\small
\begin{align}
    &\hat{R}(\mathbf{x}) = \sum_{n=1}^{N(\mathbf{x})}\int_{t_s^{n}(\mathbf{x})}^{t_e^{n}(\mathbf{x})}R(\mathbf{x},t)dt\nonumber \\ & \scalemath{.95}{=\sum_{n=1}^{N(\mathbf{x})}\int_{t_s^{n}(\mathbf{x})}^{t_e^{n}(\mathbf{x})}\log_2\left(1+\frac{\gamma l_p(f_c)G_0c^2}{\Phi_a|\phi (\mathbf{x},t)| \left(4\pi f_c\sqrt{(d_{\mathbf{x}}^2(t)+H^2)}\right)^2}\right)dt}.
\vspace{-3 cm}
\end{align}
\normalsize
As elaborated earlier, in our proposed localization scheme, accuracy depends on the total data successfully received at any location within the coverage area. According to (7), the amount of data transferred to a location $\mathbf{x}$ is exclusively dependent on the trajectory of the UAV. Therefore, in the following subsection, we formulate an optimization problem to find a suitable trajectory for the UAV such that the amount of data transferred to all the locations within the coverage area is higher than a given threshold.
\vspace{-3mm}
\subsection{Optimization Problem Formulation}
In this paper, our goal is to devise a trajectory for UAV so that the total mission completion time is minimized. With regard to the discussions in the previous subsections, the mission completion time minimization problem can be mathematically formulated as 

\textbf{P1:}
\vspace{-4mm}
\small
\begin{mini}|s|<b>
{T,\mathbf{u}(t),t_s^{n}(\mathbf{x}),t_e^{n}(\mathbf{x})}{T}{}{}
\addConstraint{C1:~ \hat{R}(\mathbf{x})\geq R_{th},}
\addConstraint{C2:~0\leq t_s^{n}(\mathbf{x})\leq T,~~ \forall ~~ \mathbf{x}\in \mathcal{Q},}
\addConstraint{C3:~0\leq t_e^{n}(\mathbf{x})\leq T,~~\forall ~~\mathbf{x}\in \mathcal{Q},}
\addConstraint{C4:~t_s^{n}(\mathbf{x})\leq t_e^{n}(\mathbf{x}),~~\forall~~\mathbf{x}\in \mathcal{Q},}
\addConstraint{C5:~\mathbf{u}(t)\in [\mathcal{R}^{1\times 2}, H], ~~ \forall~ t\in ~[0,T],}
\addConstraint{C6:~\frac{d\mathbf{u}(t)}{dt}\neq 0,~~ \forall~ t \in [0,T],}
\addConstraint{C7:~ N(\mathbf{x}) \in \mathcal{N} \setminus \{0\},}
\end{mini}
\normalsize
where $\mathcal{Q}$ denotes the set of all the points on the ground within the circular area centered at $\mathbf{c}_0$ with radius $R_0$, and $R_{th}$ is the minimum amount of data that must be received at all the locations within the scanning area. The value of $R_{th}$ can be chosen in such a way that the minimum required beacon sequences are transferred to any point $\mathbf{x} \in \mathcal{Q}$ within its connection time duration $\sum_{n=1}^{N(\mathbf{x})}t_e^n(\mathbf{x})-t_s^n(\mathbf{x})$.

In \textbf{P1}, the objective function is the completion time. Constraint C1 guarantees that the received data at all the points within the coverage area over the course of completion time is higher than a threshold so that the ToA-based positioning can be performed. Constraints C2 and C3 make sure that the start and end of the connection times for all the points within the coverage area are smaller than the total completion time. C4 ensures that the start of the connection time is smaller than the end of the connection time for all the points within the coverage area. Constraint C5 guarantees that the UAV moves on a 2D plane at height $H$ all the time. Finally, constraint C6 imposes that the UAV remains mobile for the whole duration of the mission.

Problem \textbf{P1} is challenging due to the following reasons. \textit{First}, finding simpler closed-form expressions for $\hat{R}(\mathbf{x},t)$, which can be utilized to get insights about the solution and the convexity of \textbf{P1}, is difficult. Hence, it is unclear whether \textbf{P1} is a convex optimization problem or not, leaving vagueness about any claims that can be made about the global/local optimality of the obtained solution. \textit{Second}, the number of constraints generated due to C1-C4 in \textbf{P1} is infinite owing to the continuous nature of the points in $\mathcal{Q}$. The presence of infinite number of constraints in \textbf{P1} makes it much more challenging than the communications based optimization problems considered in \cite{zeng1,zeng2,zeng3}, where a discrete set of ground users with known locations needs to be served. Although in a communications system scenario it is reasonable to assume the availability of location knowledge, this is not necessarily true in a rescue situation, where users can be located anywhere within a given area. \textit{Third}, the number of constraints generated due to C5-C6 is also infinite due to the continuous nature of the set $[0,T]$. Therefore, the total number of constraints in \textbf{P1} is infinite, rendering any solution approach based on exhaustive search computationally infeasible. Despite the above challenges posed by \textbf{P1}, we propose computationally efficient algorithms for solving \textbf{P1} in the next section.
\vspace{-4mm}

\section{Proposed Completion Time Minimization Scheme}
This section provides a low complexity solution for problem \textbf{P1} by simplifying it in two steps. A brief summary of these steps is given as follows. \textit{First}, we remove the complexity arising due to the infinite number of points within $\mathcal{Q}$ by using a lower bound on $\hat{R}(\mathbf{x},t)$. This is mainly achieved by considering the worst non-zero receiving power location within the scanning area at any given instant. \textit{Second}, we remove the complexity arising due to the infinite number of time instants within $[0,T]$ by limiting the path followed by the UAV to be circular. An illustration of the assumed UAV trajectory is provided in Fig. 1. With the circular trajectory limitation imposed, the remaining task is to find out the appropriate moving velocity during each time instant within the whole completion time and find the appropriate radius of the UAV's circular trajectory. 
\begin{figure}
\centering
  \includegraphics[width=\columnwidth]{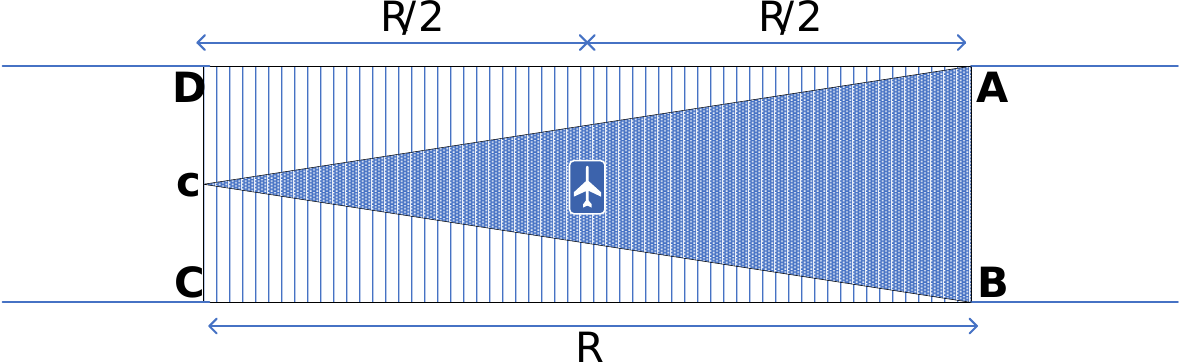}
  \caption{Top view for the coverage area of the UAV in Fig. 1 at time $t$. The angle $\angle \mathbf{A}\mathbf{c}\mathbf{B}=2\Phi_a$. Here, we assume $\mathbf{c}=[0,0,H]$ and $\mathbf{u}(t)=[\frac{R}{2},0,H]$. Area $\mathcal{A}$ is the rectangle formed by points $\mathbf{A,B,C}$ and $\mathbf{D}$ while the area $\mathcal{A}'$ is cone formed by points $\mathbf{A,c}$ and $\mathbf{B}$.}%
  \vspace{-2mm}
\end{figure}

\begin{figure}
\centering
  \includegraphics[width=0.85\columnwidth]{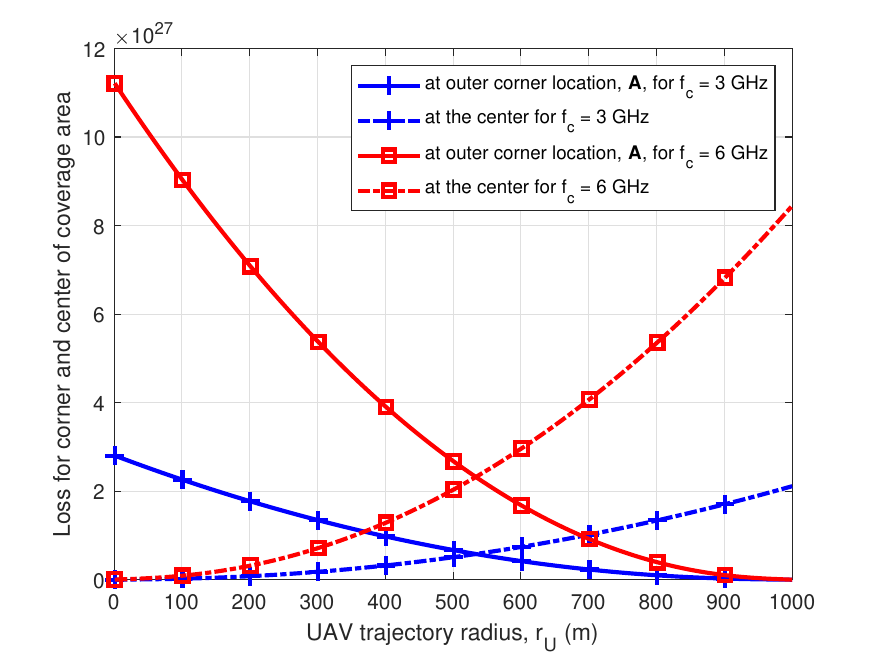}\\
  \caption{The path loss comparison for the worst point and the center point.}
  \vspace{-4mm}
\end{figure}

\vspace{-4mm}
\subsection{Full Coverage with a Single Circular Trajectory}
Without the loss of generality, consider time instant $t$, and denote the magnitude of the elevation angle corresponding to point $\mathbf{A}$ by $|\phi_e(\mathbf{A},t)|$. As illustrated in Fig. 2, due to the symmetry, the magnitude of the elevation angles corresponding to points $\mathbf{B},\mathbf{C},\mathbf{D}$ are also $|\phi_e(\mathbf{A},t)|$. Note that the antenna gain corresponding to each point within the shaded stripped area, $\mathcal{A}$, is at least $\frac{G_0}{|\phi_e(\mathbf{A},t)|\Phi_a}$. Moreover, the distance of each point within $\mathcal{A}$ from the UAV is smaller than the distance corresponding to point $\mathbf{A}$. Therefore, for any random point, $\hat{\mathbf{x}}$, within $\mathcal{A}$, we have 
\vspace{-3mm}
\begin{equation}
R(\hat{\mathbf{x}},t)\geq \log_2\left(1+\frac{\gamma l_p(f_c)G_0c^2}{|\phi_e(\mathbf{A},t)|\Phi_a \left(4\pi f_c\sqrt{(d_{\mathbf{A}}^2+H^2)}\right)^2}\right).
\end{equation}

Next, to remove the complexity arising due to $t_s(\hat{\mathbf{x}}),t_e(\hat{\mathbf{x}})$, consider the cone that is obtained by the rays that join points $\mathbf{c}$, $\mathbf{A}$, and points $\mathbf{c}$, $\mathbf{B}$. Then, our goal is to find $t_e(\hat{\mathbf{x}})-t_s(\hat{\mathbf{x}})$. Moreover, we assume that the UAV follows a circular trajectory of radius $r_U=\frac{R}{2}$ with angular velocity $v$. Then, for each point, $\hat{\mathbf{x}}$, in the conic area $\mathcal{A}' \subset \mathcal{A}$, it can be easily verified that 
\vspace{-4mm}
\begin{equation}
    t_e(\hat{\mathbf{x}})-t_s(\hat{\mathbf{x}})\geq \frac{\Phi_a}{v}, \forall~ \hat{\mathbf{x}}\in \mathcal{A}',
    \vspace{-3mm}
\end{equation}
and the lower bound on the data transferred to any point within $\mathcal{A}'$ is given as
\vspace{-4mm}
\begin{align}
\small
    & \hat{R}(\hat{\mathbf{x}})\! \geq\! \scalemath{.9}{\int_{t_s(\hat{\mathbf{x}})}^{t_e(\hat{\mathbf{x}})}\!\!\!\!\log_2\!\left(1\!+\!\frac{\gamma l_p(f_c)G_0c^2}{|\phi_e(\mathbf{A},t)|\Phi_a \left(4\pi f_c\sqrt{(d_{\mathbf{A}}^2+H^2)}\right)^2}\right)\!dt} \nonumber \\ &\geq\! \frac{\Phi_a}{v} \log_2\left(1\!+\!\frac{\gamma l_p(f_c)G_0c^2}{|\phi_e(\mathbf{A})|\Phi_a \left(4\pi f_c\sqrt{(d_{\mathbf{A}}^2+H^2)}\right)^2}\right), 
\end{align}
\normalsize
Thus, the problem \textbf{P1} is simplified to the single variable optimization problem as follows:

\textbf{P2:}
\vspace{-5mm}
\begin{maxi}|s|<b>
{v}{v}{}{}
\vspace{-8mm}
\addConstraint{\!C7:~\scalemath{.8}{\frac{\Phi_a}{v}\log_2\!\left(\!1\!+\!\frac{\gamma l_p(f_c)G_0c^2}{|\Phi_e(\mathbf{A})|\Phi_a \!\left(4\pi f_c\sqrt{(d_{\mathbf{A}}^2\!+\!H^2)}\right)^2}\!\right) \geq R_{th}}.}
\vspace{-2 cm}
\end{maxi}
Although the above problem is non-convex, it can be easily seen that the constraint function is a decreasing function of $v$. Hence, the optimal solution is obtained when the constraint is met with equality. Therefore, the optimal solution can be obtained in closed-form. In the following, we denote the optimal solution of \textbf{P2} by $v^{\textbf{P2}}$.
\vspace{-10mm}
\subsection{Optimal Radius for the Circular Trajectory with Fixed $v$}
\vspace{-4mm}
In this subsection, first we show that $\frac{R}{2}$ is not an optimal choice for $r_U$. Then, we illustrate how to find the optimal value of $r_U$. Consider the points $\mathbf{c}$ and $\mathbf{A}$ in Fig. 2, it is clear that if $r_U=\frac{R}{2}$ then $d_{\mathbf{A}}=d_{\mathbf{C}}> d_{\mathbf{c}}$ and $\Phi_e(\mathbf{A})=\Phi_e(\mathbf{C})=\Phi_e(\mathbf{c})$. This means that the received SNR at point $\mathbf{c}$ is greater than that at point $\mathbf{A}$ if we chose $r_U=\frac{R}{2}$. Moreover, if we increase the value of $r_U$ in the set $[\frac{R}{2},R]$, then the values of $d_{\mathbf{A}}$, $\Phi_e(\mathbf{A})$ decrease while $d_{\mathbf{c}}$, $\Phi_e(\mathbf{c})$ increase. Hence, the path loss by increasing $r_U$, the path loss for a point $\mathbf{A}$ decreases while for point $\mathbf{c}$ increases as shown in Fig. 3. Assuming that for some unique $ \hat{r}_U \in (\frac{R}{2},R]$ we have 
\begin{align}
|\phi_e(\hat{\mathbf{x}},\hat{r}_U)|&\Phi_a \!\left(4\pi f_c\sqrt{(d_{\hat{\mathbf{x}}}^2(\hat{r}_U)\!+\!H^2)}\right)^2 \nonumber \\
&\leq |\phi_e(\mathbf{A},\hat{r}_U)|\Phi_a \!\left(4\pi f_c\sqrt{(d_{\mathbf{A}}^2(\hat{r}_U)\!+\!H^2)}\right)^2 \nonumber \\ 
& =|\phi_e(\mathbf{c},\hat{r}_U)|\Phi_a \!\left(4\pi f_c\sqrt{(d_{\mathbf{c}_0}^2(\hat{r}_U)\!+\!H^2)}\right)^2 \nonumber \\ 
& < |\phi_e(\mathbf{A})|\Phi_a \left(4\pi f_c\sqrt{(d_{\mathbf{A}}^2+H^2)}\right)^2,~ \forall ~\hat{\mathbf{x}} \in \mathcal{A}'.
\end{align}
Then, for any fixed value of $v$ we have
\begin{align}
&\frac{\Phi_a}{v}\log_2\Bigg(\!1\!+\!\frac{\gamma l_p(f_c)G_0c^2}{|\phi_e(\hat{\mathbf{x}},\hat{r}_U)|\Phi_a \!\left(4\pi f_c\sqrt{(d_{\hat{\mathbf{x}}}^2(\hat{r}_U)\!+\!H^2)}\right)^2}\!\Bigg) \nonumber \\ & \geq \frac{\Phi_a}{v}\log_2\left(\!1\!+\!\frac{\gamma l_p(f_c)G_0c^2}{|\phi_e(\mathbf{A},\hat{r}_U)|\Phi_a \!\left(4\pi f_c\sqrt{(d_{\mathbf{A}}^2(\hat{r}_U)\!+\!H^2)}\right)^2}\!\right)\nonumber \\
    & = \frac{\Phi_a}{v}\log_2\left(\!1\!+\!\frac{\gamma l_p(f_c)G_0c^2}{|\phi_e(\mathbf{c},\hat{r}_U)|\Phi_a \!\left(4\pi f_c\sqrt{(d_{\mathbf{c}_0}^2(\hat{r}_U)\!+\!H^2)}\right)^2}\!\right) \! \nonumber \\ &> \!\frac{\Phi_a}{v}\!\log_2\!\left(\!\!1\!+\!\frac{\gamma l_p(f_c)G_0c^2}{|\phi_e(\mathbf{A})|\Phi_a \!\left(4\pi f_c\sqrt{(d_{\mathbf{A}}^2\!+\!H^2)}\right)^2}\!\right).
\end{align}
\normalsize

Hence, the corresponding optimization problem for finding optimal $r_U$ for any fixed value of $v$ can be written as the following feasibility problem.

\textbf{P3:}
\vspace{-3mm}
\begin{maxi}|s|<b>
{r_U}{0}{}{}
\vspace{-4mm}
\addConstraint{\scalemath{.8}{\frac{\Phi_a}{v}\log_2\left(1+\frac{\gamma l_p(f_c)G_0c^2}{|\phi_e(\mathbf{A},r_U)|\Phi_a\left(4\pi f_c\sqrt{(d_{\mathbf{A}}^2(r_U)+H^2)}\right)^2}\right)
\geq R_{th}},}
\addConstraint{\frac{R_0}{2}\leq r_U\leq \hat{r}_U.}
\end{maxi}
\normalsize
As the feasible set for \textbf{P3} is larger than that of \textbf{P2}, the value of $v$ can be increased even further than the optimal solution of \textbf{P2} if the optimization on $r_U$ is also performed.
\vspace{-2mm}
\subsection{Joint Optimization of $r_U$ and $v$}
With the help of the above analysis, the optimization problem for finding the optimal values of $v$ and $r_U$ can be formulated as follows

\textbf{P4}:
\vspace{-3mm}
\begin{maxi}|s|<b>
{v, r_U}{v}{}{}
\addConstraint{\scalemath{.8}{C8:\frac{\Phi_a}{v}\!\log_2\!\left(\!1\!+\!\frac{\gamma l_p(f_c)G_0c^2}{|\phi_e(\mathbf{A},r_U)|\Phi_a\left(4\pi f_c\sqrt{(d_{\mathbf{A}}^2(r_U)\!+\!H^2)}\right)^2}\!\right)
\!\geq\! R_{th}},}
\addConstraint{\scalemath{.8}{C9:~|\phi_e(\mathbf{A},r_U)|\Phi_a \!\left(4\pi f_c\sqrt{(d_{\mathbf{A}}^2(r_U)\!+\!H^2)}\right)^2}} \addConstraint{}{~~~~=|\phi_e(\mathbf{c},r_U)|\Phi_a \!\left(4\pi f_c\sqrt{(d_{\mathbf{c}}^2(r_U)\!+\!H^2)}\right)^2,}
\addConstraint{\scalemath{.8}{C10:~\frac{R}{2}\leq r_U \leq R.}}
\end{maxi}
As pointed out earlier, there is a unique value of $r_U=\hat{r}_U$ for which constraint C9 is satisfied. Therefore, \textbf{P3} can be equivalently written as:

\textbf{P5:}
\vspace{-3mm}
\begin{maxi}|s|<b>
{v}{v}{}{}
\addConstraint{\!\scalemath{.8}{C11:\frac{\Phi_a}{v}\!\log_2\!\left(\!1\!+\!\frac{\gamma l_p(f_c)G_0c^2}{|\phi_e(\mathbf{A},\hat{r}_U)|\Phi_a \!\left(4\pi f_c\sqrt{(d_{\mathbf{A}}^2(\hat{r}_U)\!+\!H^2)}\right)^2}\!\right)\! \geq\! R_{th}},}
\end{maxi}

\begin{algorithm}[t]
\caption{Finding optimal value of $r_U$ and $v$ of minimized the total completion time.}
\KwInput{$\mathbf{c}, \mathbf{A}, R$.}
\KwOutput{$r_U, v$.}
Obtain $\hat{r}_U=r_U$ through Bisection search such that the constraint C9 is satisfied\;
Obtain $v$ through Bisection search such that constraint C11 is met with equality\;
\end{algorithm}
\vspace{0mm}
Next, by using the inequality (12) and the fact that C11 must be active for optimality, we can establish that $v^{\textbf{P4}}\geq v^{\textbf{P2}}$. In \textbf{Algorithm 1}, we summarize the steps needed to obtain the optimal radius and velocity of the UAV to minimize the total completion time. Furthermore, in order to analyze the effect of $r_U$ on the optimal velocity (and subsequently the completion time), we consider removing constraint C9 from \textbf{P4} by using a fixed value of $r_U$ from the set $(\frac{R}{2}, \hat{r}_U)$. Mathematically, the modified problem can be written as

\textbf{P4-mod}-$r_U$:
\vspace{-4mm}
\begin{maxi}|s|<b>
{v}{v}{}{}
\addConstraint{~C8}
\end{maxi}
The following lemma provides a result about the optimal solutions of problem \textbf{P4-mod}-$r_U$.
\begin{lemma}
The optimal solution of \textbf{P4-mod}-$r_U$ is an increasing function of $r_U$. Mathematically, we have
\begin{equation}
    v^{\textbf{P4-mod}-r_U^1}\geq v^{\textbf{P4-mod}-r_U^2},
\end{equation}
whenever $r_U^1 \geq r_U^2$ and hence the optimal completion time is a decreasing function of $r_U$ when $r_U$ is constrained within the set $(\frac{R}{2},\hat{r}_U).$
\end{lemma}
\begin{proof}
Please see Appendix A for proof.
\end{proof}
\vspace{-2mm}
Let us denote the completion time for $r_U=\hat{r}_U$ for a specific data threshold by $\hat{T}(R_{th})$ and for $r_U =\frac{R}{2}$ by $T^{\dagger}(R_{th})$. Then, in the following proposition, we discuss the effects of finite maximum UAV velocity on $T_{sav}(R_{th})\triangleq \hat{T}(R_{th})-T^{\dagger}(R_{th})$ with respect to the data threshold $R_{th}$.
\vspace{-2mm}
\begin{proposition}
For any value of $v_{max}$, $T_{sav}(R_{th})\geq 0$ for all values of $R_{th}$. Moreover, there is at most one continuous interval of $R_{th}$ where $T_{sav}(R_{th})$ increases linearly.   
\end{proposition}
\vspace{-0.5 cm}
\begin{proof}
Please see Appendix B for proof.
\end{proof}
\vspace{-4mm}
\subsection{Adaptive Trajectory Design}
Although the scheme presented in the previous subsections achieves the objective of providing positioning service for the overall area of interest, it is not efficient. This can be explained as follows. The single circular trajectory design relies on providing sufficient SNR for the worse channel conditions anticipated at the ground users. This can increase the total completion time. The following toy example explains this. Assume that there is only one building present in the total coverage area. The scheme introduced in the previous subsection uses the channel conditions of the ground users present within this building, which is generally much poorer than the channel conditions of the users located outside the building, to find the optimal velocity for traversing the circular trajectory. This is not optimal since it will require the UAV to traverse the whole circular trajectory at a much lower speed in order to meet higher association time demands due to the poorer channel conditions used for the reference indoor ground user. A better design would be to traverse the area covered by the building at a lower speed so that the indoor ground users can receive a sufficient number of packets for their own positioning while traversing the rest of the area at a higher velocity.

Now, we proceed to devise a flexible trajectory design that considers building locations. To this end, the area occupied by the buildings is approximated by the smaller circular coverage areas. We denote the center and radius of each of the smaller circular areas by $\mathbf{c}_i$ and $r_i$, respectively. Then, the trajectory of the UAV is divided into two phases. Specifically, during the \textit{first} phase, the UAV traverses the whole coverage area by a circular trajectory. During this travel time, the goal is to provide positioning service for the outdoor ground users in the coverage area. During the \textit{second} phase, the UAV sequentially provides positioning service to the ground users in each of the smaller circular coverage areas by traversing a circular trajectory. A pictorial representation of the proposed design with only four smaller circular areas is illustrated in Fig. 4.

\begin{figure}[t]
\centering
\includegraphics[width=\columnwidth]{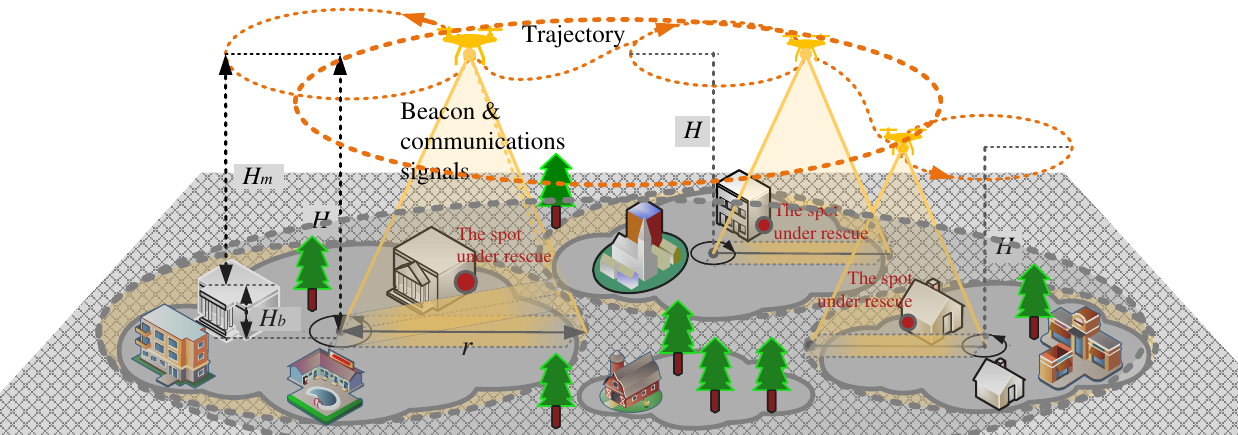}
\caption{Illustration of the trajectory design for the second phase with only four circular areas to be covered. The outer ground circles are the areas that need coverage. The elevated circles are the circular trajectory traversed by the UAV to provide positioning and communication services to the ground users located inside the corresponding circular regions. The curved lines are the trajectories traversed by the UAV after completing one black circular trajectory and before starting the next black circular.}
\vspace{-3mm}
\end{figure}    

The problem at hand is to devise the complete trajectory for the second phase so that the overall travel time is minimized while guaranteeing sufficient association time for each ground user in the coverage area. Note that the total time strongly depends on the UAV's order of visit at each smaller circular coverage area. Designing a globally optimal order of visits is a challenging problem. To understand this, let us assume that $\lim r_i\to 0, ~\forall~ i \in ~\{1,\cdots, I\}$. Then it is clear that the optimal trajectory design is at least as complex as the Traveling Salesman Problem (TSP).

\subsubsection{The Order of Visits of Areas}
For a set of locations and the cost of travel or distance between each possible pairs, the objective in TSP is to find the best possible route of visiting all the locations and returning to the starting point that minimizes the travel cost or travel distance. It can be shown that TSP is an NP-hard problem and the complexity of the solution algorithm grows with the number of locations. Although TSP is an NP-hard problem, several efficient algorithms have been proposed in the literature. Owing to the similarity between the TSP and the problem in this paper, we choose to use the order of visits to individual smaller circular trajectories obtained by solving the TSP with locations $\mathbf{c}_i$. Furthermore, without loss of generality, we denote the optimal order of visits obtained by solving TSP is given as 
\vspace{-2mm}
\begin{equation}
    \mathbf{c}^*_1 \to \mathbf{c}^*_2\to \cdots \to \mathbf{c}^*_I,
\end{equation}
where $\mathbf{c}_i\to \mathbf{c}_j$ means that the UAV first visits the circular area centered at $\mathbf{c}_i$ and then visits the area centered at $\mathbf{c}_j$.  
\subsubsection{Proposed Algorithm for a Given Order of Visits}

Having obtained the order of visits to the individual circular trajectories, the goal is to find the appropriate connecting points on the circular trajectories so that the overall time is minimized. Before proceeding further, we present an important result in the following lemma which will be helpful in our later discussions.

\begin{lemma}
For any given values of $\mathbf{c}_i\in \mathbb{R}^2,r_i$, and the order of visits to the smaller circular coverage areas, the total size of the travelled distance by the UAV is minimized when the individual circular trajectories associated with each $\mathbf{c}_i$ are connected with straight line segments.
\end{lemma}
\vspace{-0.5 cm}
\begin{proof}
Please see Appendix C for proof.
\end{proof}
\vspace{-0.25 cm}
Lemma 2 establishes that any consecutive circular trajectories of the UAV must be connected with the shortest path possible. However, it does not elaborate on the optimal starting, $\mathbf{p}_s^i$, and leaving points, $\mathbf{p}_e^i$, on the circular trajectory corresponding to $\mathbf{c}_i$. The following lemma settles this argument.

\begin{lemma}
The optimal starting and leaving points on the circular trajectory corresponding to $\mathbf{c}_i$ are same.
\end{lemma}
\vspace{-0.5 cm}
\begin{proof}
Please see Appendix D for proof.
\end{proof}
\vspace{-2mm}
Lemma 3 shows that all the $m_i$'s, where $m_i$ is defined as the length of the straight line path taken by the UAV for travelling between $\bold{p}_s^i$ and $\bold{p}^i_e$, must be zero in the optimal solution. Therefore, the overall length of the path can be expressed as
\vspace{-2mm}
\begin{equation}
    L=\sum_{i=1}^IC_i+\sum_{i=1}^{I-1}l_i.
\end{equation}

Recall that once the values of $r_i$'s are fixed, the time needed to cover the corresponding circular trajectories is also fixed. Hence, the only option available for minimizing the overall completion time is to minimize $\sum_{i=1}^{I-1}l_i$. In other words the question at hand can be written as, \textit{how to connect the circular trajectories such that $\sum_{i=1}^{I-1}l_i$ is minimized?}

As pointed out earlier, the problem of minimizing $\sum_{i=1}^{I-1}l_i$ becomes a TSP problem in the limit when $r_i\to 0$. Therefore, in the following, we propose an iterative algorithm which can be proved to obtain a non-increasing objective value of $\sum_{i=1}^{I-1}l_i$ in each successive iteration. Thus, the convergence of the proposed algorithm is guaranteed.

Before presenting the technical details of the proposed algorithm, we describe the general idea behind the proposed approach. In the proposed algorithm, at the $i$-th iteration, our goal is to minimize the distance traveled for connecting three circular trajectories of fixed radii, $r_{i-1}, r_{i}, r_{i+1}$ centered at locations $\mathbf{c}_{i-1}, \mathbf{c}_i, \mathbf{c}_{i+1}$, respectively. During each iteration, the minimization is achieved only by optimizing the connecting point on the middle circle, that is $i$-th circle, while the connecting points on all the remaining circles are left unchanged. Thus, the rest of the trajectory lengths are not altered. By doing so, only $l_i+l_{i+1}$ is minimized while the rest of the sum in $\sum_{i=1}^{I-1}l_i$ remains unaltered. Therefore, in each iteration we achieve a lower value of $\sum_{i=1}^{I-1}l_i$ as compared to the previous iteration. This, thus guarantees the convergence of the proposed iterative algorithm.

Without loss of generality, let us assume that the order of visits suggested by TSP algorithm is such that the circular trajectory corresponding to $\mathbf{c}_i$ is visited before that of the $\mathbf{c}_j$, whenever $i<j$. Furthermore, denote by $\mathbf{n}_i^j$ as the closest point on the $i$-th circular trajectory to the end point on the $j$-th circular trajectory. Then, we have the following lemma.

\begin{lemma}
The optimal connecting point, that minimizes $l_i+l_{i+1}$, on the $i$-th circular trajectory lies on the smaller circular arc that joins the points $\mathbf{n}_i^{i-1}$ and $\mathbf{n}_i^{i+1}$.
\end{lemma}
\vspace{-0.5 cm}
\begin{proof}
Please see Appendix E for proof.
\end{proof}
\vspace{-0.25 cm}
Lemma 4 only indicates the range of points on a fixed radius arc where the optimal connecting point can lie. However, it does not elaborate on the approach that can be used to find it. In the following lemmas, we settle this argument.

\begin{lemma}
For a fixed radius of the $i$-th circular trajectory, and $\mathbf{p}_e^{1},\cdots, \mathbf{p}_e^{i-1},\mathbf{p}_e^{i+1},\cdots, \mathbf{p}_e^{I}$, the optimal connecting point on the $i$-th circular trajectory, $\mathbf{p}_e^i$, can be found through performing Bisection search over the points on the smaller arc that joins the points $\mathbf{n}_i^{i-1}$ and $\mathbf{n}_i^{i+1}$.
\end{lemma}
\vspace{-0.5 cm}
\begin{proof}
Please see Appendix F for proof.
\end{proof}
\vspace{-0.25 cm}
Lemma 5 provides the optimal angle of $\mathbf{p}_e^i$ for a fixed radius of the $i$-th circular trajectory. In the following lemma, we show how to obtain the optimal value of $r_i$ for a fixed value of $\angle{\mathbf{p}_e^{i}}=\phi^*$.
\begin{lemma}
The optimal value of $r_i^{j+1}$ is obtained via performing Bisection search over $[r_i^j, r_{i}^{opt}]$. 
\end{lemma}
\vspace{-0.5 cm}
\begin{proof}
Please see Appendix G for proof.
\end{proof}
\vspace{-0.25 cm}
Based on Lemma 5 and Lemma 6 we develop an iterative algorithm that finds the optimal value of $\mathbf{p}_e^{i}$. The convergence of the proposed iterative algorithm relies on the non-increasing property of the objective function achieved in each successive iteration. In the following lemma, we show the \textbf{Algorithm 2} achieves a non-increasing value of the completion time.

\begin{lemma}
For fixed values of $\mathbf{p}_e^{1},\cdots, \mathbf{p}_e^{i-1},\mathbf{p}_e^{i+1},\cdots, \mathbf{p}_e^{I}$, the optimal value of $\mathbf{p}_e^i$ provided by \textbf{Algorithm 2} achieves a non-increasing total completion time in successive iterations.
\end{lemma}
\vspace{-0.5 cm}
\begin{proof}
Please see Appendix H for proof.
\end{proof}

\begin{algorithm}[t]
\caption{Finding optimal value of $\mathbf{p}_e^{i}$ through iterative optimization of $\phi$ and $r_i$.}
\KwInput{$\mathbf{p}_e^{i-1}, \mathbf{p}_e^{i+1}$, $r_i^0=\frac{R_i}{2}, r_i^{opt}, \mathbf{n}_e^{i-1}, \mathbf{n}_e^{i+1}$, maximum iterations $=U$.}
\KwOutput{$\mathbf{p}_e^{i}$.}
\For{$j=1:U$}{
Solve (10) for a given value of $r_i=r_i^{j-1}$ and obtain $\phi^*$\;
Solve (12) for a given value of $\phi^*$ and obtain $r_i^{j*}$\;
}
\end{algorithm}
\vspace{-0.25 cm}
The above analysis of optimally connecting three circular trajectories leads to an iterative optimization algorithm that minimizes the total completion time. The proposed iterative algorithm is given as \textbf{Algorithm 3} and its convergence guarantee is proven in the following proposition.

\begin{proposition}
The total completion time achieved in successive iterations of \textbf{Algorithm 3} is non-increasing and therefore \textbf{Algorithm 3} is guaranteed to converge.
\end{proposition}
\vspace{-0.5 cm}
\begin{proof}
Please see Appendix I for proof.
\end{proof}

\begin{algorithm}[t]
\caption{Iterative algorithm for finding optimal values of $\mathbf{p}_e^{i}, ~\forall~ i\in \{1,\cdots, I\}$ .}
\KwInput{$\mathbf{c}_i,R_i,~ \forall ~ i \in \{1,\cdots, I\}$, $\mathbf{n}_i^{i-1}, ~\forall ~ i \in \{1,\cdots, I-1\}$, maximum iterations $=K$\;}
\KwOutput{$\mathbf{p}_e^{i*}, ~\forall i~\in \{1,\cdots, I\}$\;}
\For{$k=1:K$}{
\If{$k==1$}{
$\mathbf{p}_e^{i,I}=\mathbf{n}_i^{i-1}$\;
}
\For{$i=1:I$}{
Apply \textbf{Algorithm 1} to obtain $\mathbf{p}_e^i$ with $\mathbf{p}_e^{i-1}=\mathbf{p}_e^{i-1, kI+i-1}, \mathbf{p}_e^{i+1}=\mathbf{p}_e^{i+1,kI+i-1}$\;
Set $\mathbf{p}_e^{i,kI+i}=\mathbf{p}_e^i$\;
Set $\mathbf{p}_e^{\hat{i},kI+i}=\mathbf{p}_e^{\hat{i},kI+i-1},~ \forall ~ \hat{i}\neq i, \hat{i} \in\{1, \cdots, I\}$\;
}
}
Set $\mathbf{p}_e^{i*}=\mathbf{p}_e^{i,(K+1)I}, ~\forall i~\in \{1,\cdots, I\}$\;
\end{algorithm}

\vspace{-0.5 cm}
\section{Numerical Results}
This section presents the numerical results for the proposed algorithms. Unless specified otherwise, the simulation parameters used to obtain the results are presented in Table II. Note that due to the mobility constraint (C6 in \textbf{P1}) and the infinite number of constraints due to the uncertainty of location within the service area, the formulated UAV routing problem is radically different from those formulated in the past for communications system design. Therefore, the routing algorithms proposed to solve those problems cannot be used to obtain the solutions for the localization problem considered in this work.

This section is divided into three logical subsections. The first subsection discusses the single circular trajectory scenario with $v_{opt}\leq v_{max}$. The second subsection discusses a single circular trajectory with $v_{opt}\geq v_{max}$. Finally, the third subsection illustrates results for multiple circular trajectory cases.  
\begin{table}
\centering
\caption{Simulation parameters.}
\begin{tabular}{ |c|c|c|c|}
 \hline
 Parameter & value & Parameter & value \\
 \hline
 \hline
 $f_c$ & $\{3, 6\}$ GHz & Transmit power & $20$ dBm \\
 \hline
 $\Phi_a$ & $\{\frac{\pi}{6},\frac{\pi}{4},\frac{\pi}{3}\}$ rad & Noise density & $-174$dBm/Hz \\ 
 \hline
 $\Delta f$ & $20$ MHz & $H$ & $100$ m \\
 \hline
 $R_0$ & $1000$ m & $v_{max}$ & $72$ mph \\
 \hline
\end{tabular}
\vspace{-4mm}
\end{table}

\vspace{-6mm}
\subsection{Single Scanned Area with $v_{opt}\leq v_{max}$}
Fig. 5 shows the optimal velocity results with respect to data thresholds. It can be observed that the optimal velocity for $r_U=r_{opt}$ is higher than that for $r_U=\frac{R}{2}$. This is due to the fact that for $r_U=r_{opt}$ the path loss for the worst location in the scanning area is less as compared to $r_U=\frac{R}{2}$, thus allowing a higher speed for the UAV while satisfying the data threshold requirements. Furthermore, it can be observed that for a higher carrier frequency, the optimal velocity is smaller as compared to that for a smaller carrier frequency. This is due to the higher penetration loss, which results in overall higher path loss, and subsequently, more dwelling time is required for each location to satisfy the data threshold requirement.

\begin{figure}[t]
  \includegraphics[width=\columnwidth]{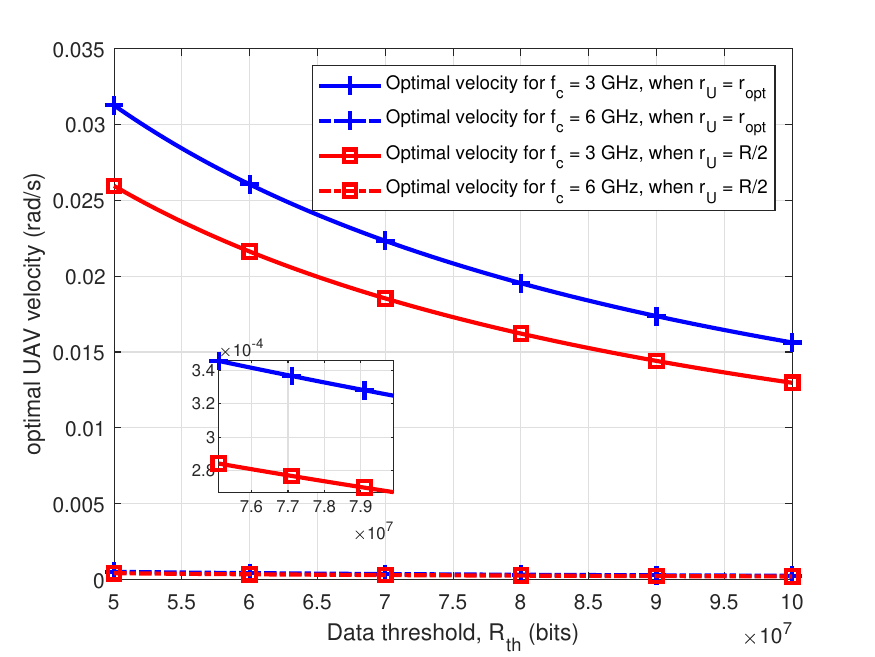}
  \vspace{-2mm}
  \caption{Optimal UAV velocity for different data thresholds with $\Phi_a=\frac{\pi}{6}$.}
  \label{fig:test1_1}
  \vspace{-4mm}
\end{figure}
\begin{figure}[t]
  \includegraphics[width=\columnwidth]{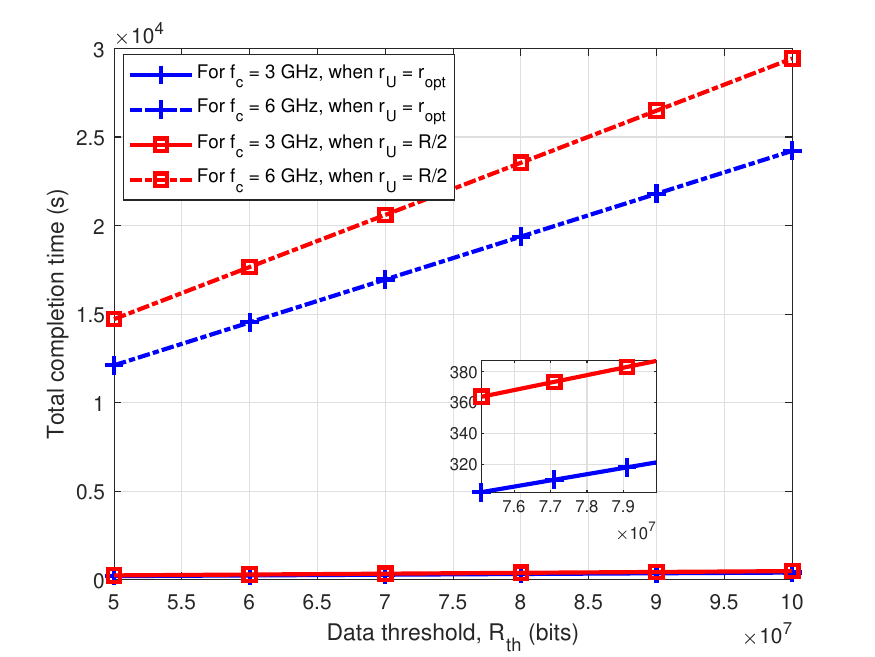}
  \vspace{-4mm}
  \caption{Total completion time for different data thresholds with $\Phi_a=\frac{\pi}{6}$.}
  \label{fig:test2_1}
  \vspace{-4mm}
\end{figure}

Fig. 6 and Fig. 7 present the total completion time and time savings percentages, for the $r_U=r_{opt}$ and $r_U=\frac{R}{2}$. For both values of $r_U$, it can be observed from Fig. 6 that the total completion time increases with the increase in data threshold requirement. This is due to the fact that a higher value of data transfer to any particular location requires more dwelling time which results in a higher value of total completion time. Moreover, the considered scenario to obtain these results require $v_{opt}$ to be less than $v_{max}$. Thus, we observe that the percentage of time savings is independent of the data threshold in Fig. 7. This result verifies the theoretical observations made in Proposition 1. The effects of having $v_{opt}\geq v_{max}$ are discussed in the following subsection.

\begin{figure}
  \includegraphics[width=\columnwidth]{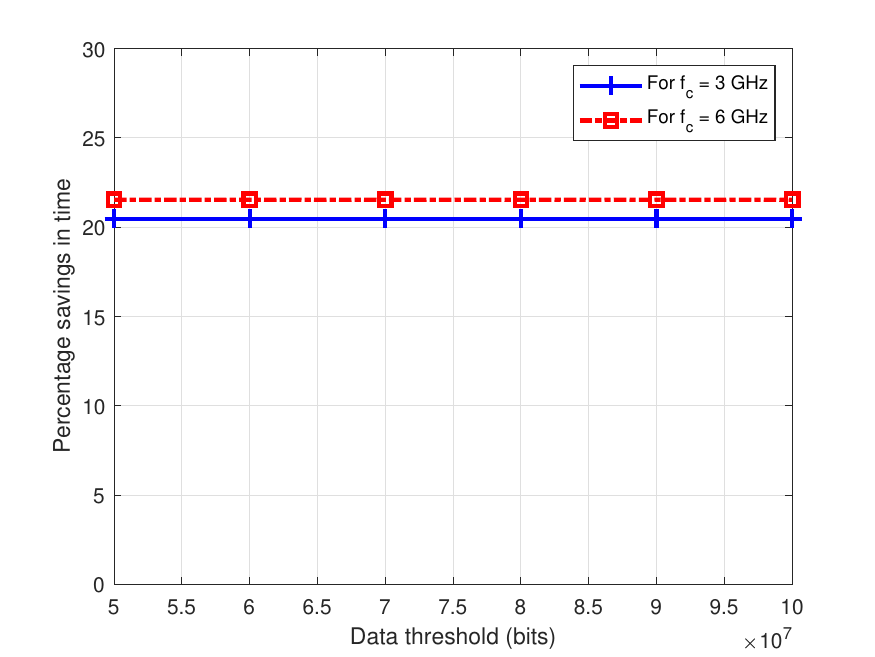}
  \vspace{-4mm}
  \caption{Percentage savings in completion time for different values of $f_c$.}
  \label{fig:test1_2}
  \vspace{-8mm}
\end{figure}

\begin{figure}
  \includegraphics[width=\columnwidth]{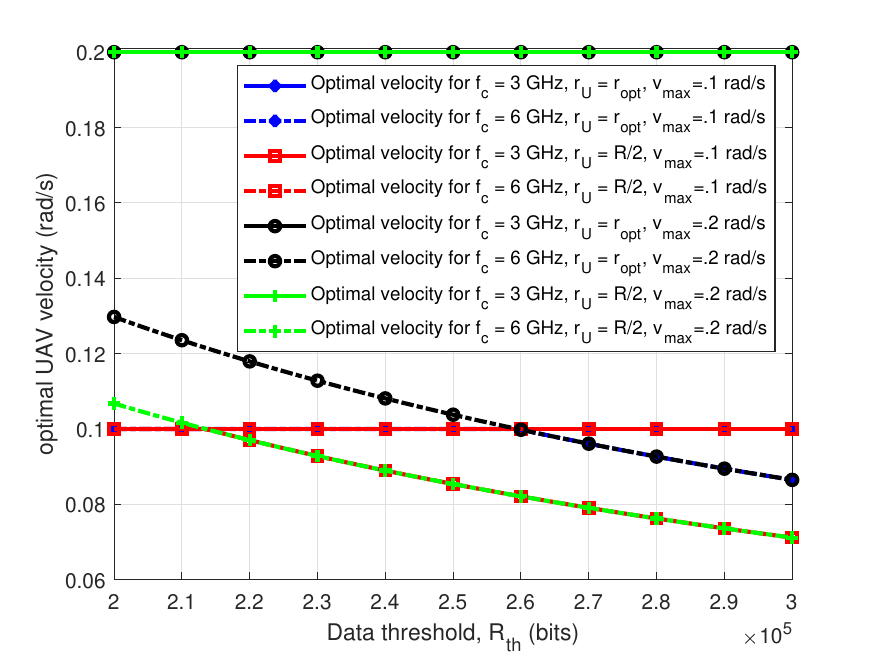}
  \vspace{-4mm}
  \caption{Optimal UAV velocity for different data thresholds with $R=300$, $\Phi_a=\frac{\pi}{6}$ m.}
  \label{fig:test2_2}
  \vspace{-1mm}
\end{figure}

\vspace{-3mm}
\subsection{Single Scanned Area with $v_{opt}\geq v_{max}$}
\vspace{-1mm}
The optimal velocity, total completion time and time savings percentage results for the case when $v_{opt} \geq v_{max}$ are provided in Fig. 8, Fig. 9 and Fig. 10, respectively. For these results, we assumed $R=300$ m. As noted in Proposition 1, there is a range of data thresholds for which the optimal velocity for both $r_U=\hat{r}_U$ and $r_U=\frac{R}{2}$ will be $v_{max}$ and for this range of data thresholds, the total completion time will be identical and time savings will be zero. In Fig. 8, we can see that for data thresholds range $[2.1, 2.6] \times 10^5$ the optimal velocity for $r_U=\frac{R}{2}$ is decreasing while that for $r_U=\hat{r}_U$ is still $=.1$ rps (the maximum possible velocity). Due to this behaviour, we see a slope between $[2.1,2.6] \times 10^5$ in the time savings graph in Fig. 10. On the other hand, for data thresholds range between $[2,2.1]\times 10^5$, we observe zero savings in total completion time since the optimal velocity for both choices of $r_U$ is same. Moreover, for data thresholds range $\geq 2.6 \times 10^5$, the time savings percentage becomes independent of the data threshold value. This also verifies the theoretical result derived earlier in Proposition 1.

\begin{figure}[t]
  \includegraphics[width=\columnwidth]{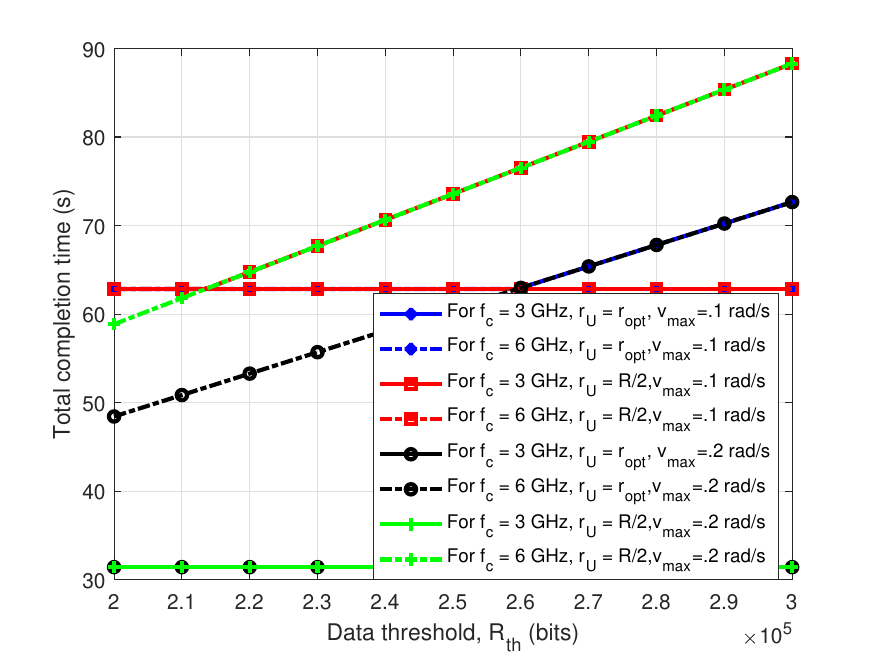}
  \caption{Total completion time for different data thresholds with $R=300$ m.}
  \label{fig:test1_3}
  \vspace{-4mm}
\end{figure}%
\begin{figure}[t]
  \includegraphics[width=\columnwidth]{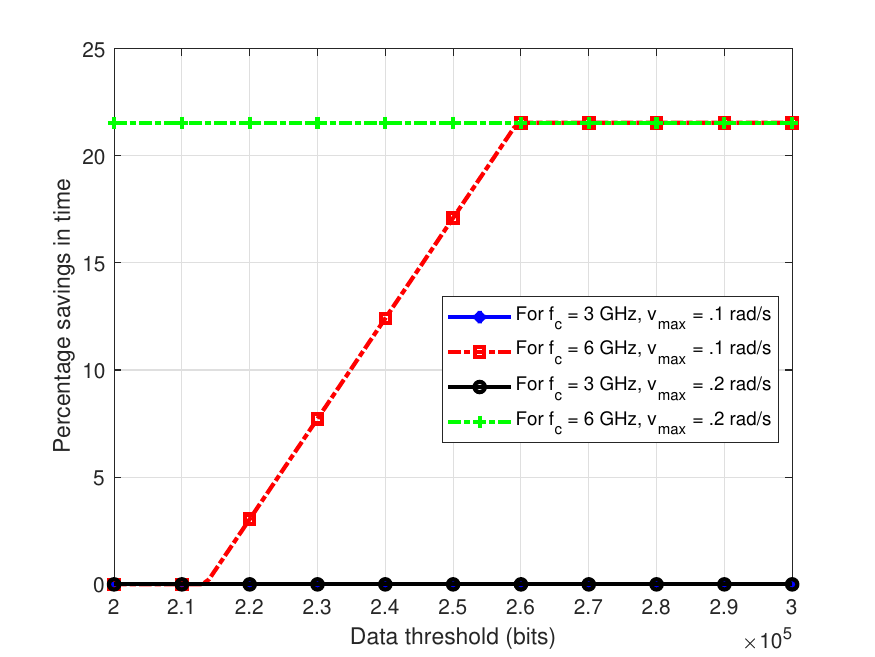}
  \caption{Percentage of savings in total completion time for different values of $f_c$ with $R=300$ m.}
  \label{fig:test2_3}
  \vspace{-4mm}
\end{figure}

\vspace{-0.5 cm}
\subsection{Multiple Scanned Areas}
For multiple circular trajectory results, we assume there are three isolated circular regions which are centered at $\mathbf{c}_1= [-353.6,0],\mathbf{c}_2= [176.8,-306.2]$ and $\mathbf{c}_3=[176.8,306.2]$, with coverage radius $r_1=200$ m, $r_2=100$ m and $r_3=200$ m. These regions are assumed to be covered with buildings. Hence, any users present within these smaller circular coverage areas are assumed to experience penetration loss. Moreover, the bigger circular area is assumed to be centered at $[0,0]$ and has a radius of $1000$ m. Apart from the smaller circular areas within the bigger circular area, the rest of the area within the bigger circular area is assumed to be covered by trees or other vegetation. Thus, any users present within this area is expected to experience much less path loss as compared to the users present within the smaller circular regions. With these settings, the optimal trajectory for the UAV is shown in Fig. 11 for $\Phi_a \in \{\frac{\pi}{4},\frac{\pi}{3}\}$. It can be observed that the optimal radius $r_{opt}$ for a smaller value of $\Phi_a$ is smaller as compared to that for higher value of $\Phi_a$. This is due to the fact that for a higher value of $\Phi_a$ the worst point (point $\mathbf{A}$ in Fig. 2) becomes further away from the UAV and hence to balance the path loss between the points $\mathbf{c}_0$ and $\mathbf{A}$ the UAV has to become closer to point $\mathbf{A}$ by traversing at a bigger radius. Hence, we observe a higher value of optimal radius for higher value of $\Phi_a$. 

Fig. 12 shows the total completion time for the multiple circular trajectory scenario. We assumed that the disparity between the path loss between for any location within building regions (smaller circular regions) and region without the buildings is $30$ dB. Thus, any location that is not within the smaller circular region will experience $30$ dB less loss. It can be observed that the proposed scheme with $r_U=r_{opt}$ performs much better than the $r_U=\frac{R}{2}$ case. Moreover, comparing total time completion results with those obtained for a single circular trajectory case (shown in Fig. 6), we can see that the multiple circular trajectory case performs much better than the single circular trajectory case. For instance, for a data threshold value of $50$ Mb, the single circular trajectory scheme requires around $12500$ seconds while the multiple circular trajectory scheme requires less than $1500$ seconds. Hence, a saving of at least $7$ times in total completion time.

\begin{figure}[t]
  \includegraphics[width=\columnwidth]{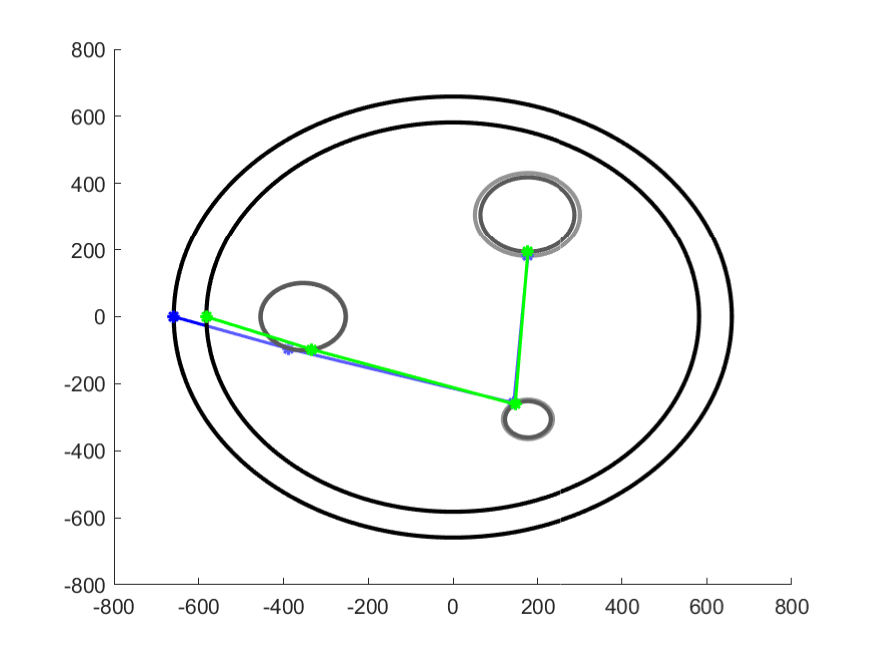}
  \caption{UAV trajectory for multiple circular trajectory case with $\Phi_a=\frac{\pi}{4}$ (green line) and $\Phi_a=\frac{\pi}{3}$ (blue line).}
  \label{fig:test1_4}
  \vspace{-5mm}
\end{figure}%
\begin{figure}
  \includegraphics[width=\columnwidth]{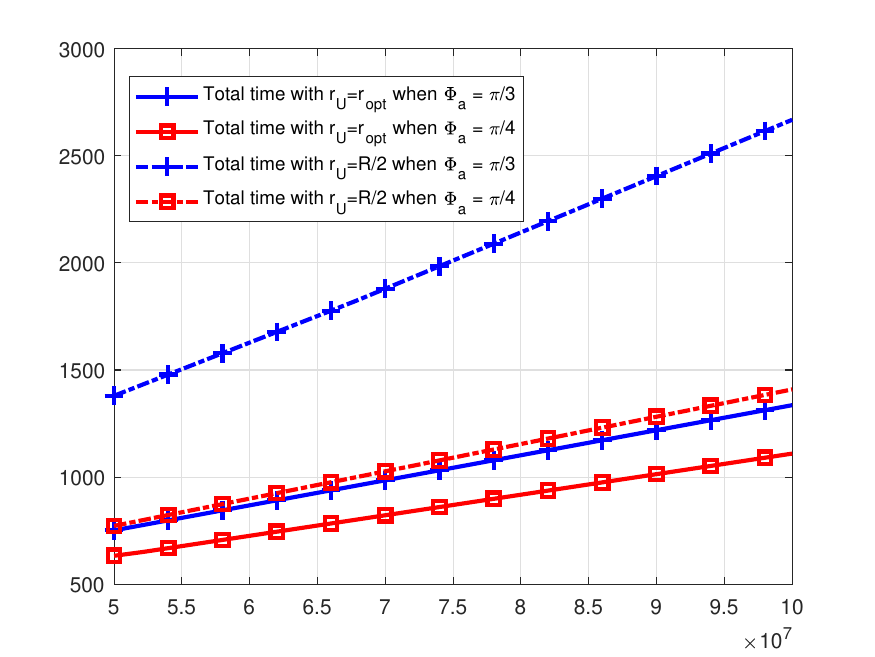}
  \vspace{-4mm}
  \caption{Total completion time for multiple circular trajectory case.}
  \label{fig:test2_4}
  \vspace{-4mm}
\end{figure}

\vspace{-2mm}
\section{Conclusions}
This paper provides new low complexity UAV trajectory design algorithms for accomplishing the positioning and communication tasks with minimum completion time in the disaster hit areas. Specifically, based on the topographical scenario, two UAV trajectory designs are proposed. In the \textit{first} design, the whole scanning area is covered by the UAV with a single circular trajectory. Then, the proposed algorithm finds the optimal UAV velocity and the optimal radius for the circular trajectory. This first algorithm uses low-complexity Bisection search for finding the optimal UAV velocity and it is applicable to the scanning areas which have uniform topography. To deal with the diversity in topography of the coverage area, we propose the \textit{second} trajectory design, where the whole scanning area is divided into smaller circular areas with uniform topography. Then, the proposed iterative algorithm finds the optimal trajectory for traversing the whole scanning area by optimally connecting the small circular trajectories. Both, theoretical and simulation results, suggest that the achieved time saving due to the application of proposed algorithms is independent of the data thresholds and it improves with the increase in carrier frequency. Since positioning performance is directly proportional to the carrier frequency, we concluded that the proposed algorithms are more applicable to high precision wideband integrated sensing and communication systems.
\vspace{-2mm}
\appendix
\vspace{-2mm}
\subsection{Proof of Lemma 1:}
\vspace{-2mm}
\begin{proof}
This lemma can be proved by first noting that the constraint function in C8 is an increasing function of $r_U, \forall r_U \in (\frac{R}{2},\hat{r}_U)$. Therefore, for any fixed value of $r_U^1>r_U^2$ we have
\begin{align}
    &\scalemath{.8}{\frac{\Phi_a}{v}\log_2\left(\!1\!+\!\frac{\gamma l_p(f_c)G_0c^2}{|\phi_e(\mathbf{A},r_U^1)|\Phi_a \!\left(4\pi f_c\sqrt{(d_{\mathbf{A}}^2(r_U^1)\!+\!H^2)}\right)^2}\!\right)}\nonumber \\ & \scalemath{.8}{> \frac{\Phi_a}{v}\log_2\left(\!1\!+\!\frac{\gamma l_p(f_c)G_0c^2}{|\phi_e(\mathbf{A},r_U^2)|\Phi_a \!\left(4\pi f_c\sqrt{(d_{\mathbf{A}}^2(r_U^2)\!+\!H^2)}\right)^2}\!\right).} 
\end{align}
\normalsize
Second, we note that the optimal solution for \textbf{P4-mod}-$r_U$ is achieved when C8 is met with equality. Therefore, if the solution of \textbf{P4-mod}-$r_U^2$ is $v^{\textbf{P4-mod}-r_U^2}$, then we cannot have $v^{\textbf{P4-mod}-r_U^1} \in[0,  v^{\textbf{P4-mod}-r_U^2}]$ since no such value can meet constraint C8 with equality. Therefore, we must have
\begin{equation}
    v^{\textbf{P4-mod}-r_U^1}\geq v^{\textbf{P4-mod}-r_U^2}.
\end{equation} This completes the proof.
\end{proof}
\vspace{-6mm}
\subsection{Proof of Proposition 1:}
\begin{proof}
First, we note the following facts:
\begin{itemize}
    \item F1: The channel gain (and subsequently SNR) for the worst point is higher if we use $r_U=\hat{r}_U$ rather than $r_U=\frac{R}{2}$.
    \item F2: After choosing a value of $r_U$, the speed should be chosen in such a way that the data constraint is met with equality.
    \item F3: The speed and data transferred have an inverse relationship for any fixed value of the radius of the trajectory.
\end{itemize}

These facts are discussed in detail in Section III-B, and Section III-C. Moreover, let us assume that at any particular time instant, the spectral efficiency for $r_U=\hat{r}_U$ is denoted by $a_1$ and that for radius $r_U=\frac{R}{2}$ by $a_2$\footnote{It can be seen that $a_1,a_2$ are independent of data threshold.}. Now according to F1, we must have 
\vspace{-2mm}
\begin{equation}
    a_1 > a_2.
\end{equation}

Next, assume that the total of data to be transferred is $R_{th}$. Then, according to F2 we have
\vspace{-2mm}
\begin{equation}
    \frac{\Phi_a a_1}{v_1} = R_{th},~~\frac{\Phi_a a_2}{v_2} = R_{th},
\end{equation}
where $v_1,v_2$ are the optimal velocities for the UAV when the trajectory radius is $\hat{r}_U$ and $\frac{R}{2}$, respectively. According to F3, we must have
\vspace{-3mm}
\begin{equation}
    v_1 > v_2.    
\end{equation}

Since the completion time is given as $\frac{2 \pi}{v}$, we must have
\vspace{-2mm}
\begin{equation}
    T_{sav}(R_{th}) = \frac{2 \pi}{v_2} - \frac{2 \pi}{v_1} > 0.
\end{equation}
However, if we put an upper limit on the value of $v$, we may have $T_2-T_1=0$ for a certain range of data. Since for such a range of data thresholds, the UAV can increase its speed for both possible values of radii but cannot do so due to the limitation on maximum velocity. This behaviour will result in no savings in time for smaller total data requirements. Now there will be some minimum value of data threshold for which the optimal speed for $r_U=\frac{R}{2}$ will be smaller than the maximum possible speed while it will be the maximum speed for trajectory radius $\hat{r}_U$, and this behaviour will continue for some range of data thresholds. Over this range of data thresholds, we should expect the time needed for $r_U=\frac{R}{2}$ to increase while remaining constant for trajectory radius $\hat{r}_U$. Therefore, we should expect a slope in the time savings curve for this range of data thresholds.
\end{proof}
\vspace{-6mm}
\subsection{Proof for Lemma 2:}
\begin{proof}
\vspace{-0.15 cm}
This lemma can be proved by contradiction. Denote the distance of the circular trajectory associated with the center $\mathbf{c}_i$ by $C_i$ and the distance of the path connecting the circular trajectories for $\mathbf{c}_i$, $\mathbf{c}_{i+1}$ by $l_i$. Moreover, denote by $\mathbf{p}_s^i, \mathbf{p}_e^i$ the points where the UAV meets and leaves the circular trajectory corresponding to $\mathbf{c}_i$, and $m_i$ denotes the length of the straight path taken by the UAV for travelling between $\mathbf{p}_s^i$ and $\mathbf{p}_e^i$. Then, the total size of the travelled distance, denoted by $L$, is equal to the sum of $C_i$'s, $l_i$'s and $m_i$'s. Mathematically, it can be written as
\begin{equation}
    L=\sum_{i=1}^{I}C_i+\sum_{i=1}^{I-1}l_i+\sum_{i=1}^{I-1}m_i.
\end{equation}

Assume, without loss of generality, that the UAV leaves the circular trajectory corresponding to $\mathbf{c}_i$ at point $\mathbf{p}_i\in \mathbb{R}^2$ and connects with the circular trajectory corresponding to $\mathbf{c}_{i+1}$ at point $\mathbf{p}_{i+1}\in \mathbb{R}^2$. Furthermore, assume that the UAV follows a non-straight path between points $\mathbf{p}_i, \mathbf{p}_{i+1}$, then we can always replace the non-straight path between points $\mathbf{p}_e^i, \mathbf{p}_s^{i+1}$ with a straight line to reduce the travel distance of UAV for any fixed values of $C_i, C_{i+1}, m_i$. This completes the proof.
\end{proof}
\vspace{-6mm}
\subsection{Proof of Lemma 3:}
\begin{proof}
Consider, without loss of generality, that the optimal values of $\mathbf{p}_s^i$, $\mathbf{p}_e^{i}$ are given by $\mathbf{p}_s^{i*}$, $\mathbf{p}_e^{i*}$ and the corresponding optimal values of $C_i, l_i, m_i$ by $C_i^*, l_i^*, m_i^*$. Furthermore, consider the points $ \mathbf{p}_s^{i+1*}, \mathbf{p}_e^{i+1*}$. Then, there are following two possibilities: \textbf{Case 1}: Points $\mathbf{p}_e^{i*}, \mathbf{p}_s^{i+1*}, \mathbf{p}_e^{i+1*}$ lie on the same line segment; \textbf{Case 2}: Points $\mathbf{p}_e^{i*}, \mathbf{p}_s^{i+1*}, \mathbf{p}_e^{i+1*}$ do not lie on the same line segment.

Case 1: For this case, we can show that for any fixed values of $C_i, l_i \forall i\in \{1,\cdots, I\}$, the optimal value of $m_{i+1}$ is attained when, after completing the circular trajectory $C_{i+1}$, the UAV travels in a straight line from point $\mathbf{p}_s^{i+1*}$ to point $\mathbf{p}_e^{i+1*}$. Moreover, the value of $L$ is non-increasing if we replace $\mathbf{p}_s^{i+1*}$ by $\mathbf{p}_e^{i+1*}$. Hence, the proof is complete for Case 1.

Case 2: First, observe that the optimal value of $m_{i+1}$ is achieved when UAV travels in a straight line between points $\mathbf{p}_s^{i+1*}$ and $\mathbf{p}_e^{i+1*}$. Using the triangle inequality, it can be easily shown that
\begin{equation}
l_i+m_{i+1}=\|\mathbf{p}_s^{i+1*}-\mathbf{p}_e^{i*}\|+\|\mathbf{p}_e^{i+1*}-\mathbf{p}_s^{i+1*}\|\geq \|\mathbf{p}_e^{i+1*}-\mathbf{p}_e^{i*}\|,
\end{equation}
the equality holds when $\mathbf{p}_s^{i+1*}=\mathbf{p}_e^{i+1*}$. The UAV must use $\mathbf{p}_s^{i+1*}=\mathbf{p}_e^{i+1*}$ to minimize $L$.
\end{proof}
\vspace{-6mm}
\subsection{Proof of Lemma 4:}
\begin{proof}
Assume, without loss of generality, that $\mathbf{c}_i=[0, 0]$ and the line joining the $\mathbf{c_i}, \mathbf{p}_e^{i-1}$ is the x-axis. Then, we can have four possibilities for $\mathbf{p}_e^{i+1*}$. Specifically, \textbf{Case 3}: $\mathbf{p}_e^{i+1*}$ lies in the first quadrant; \textbf{Case 4}: $\mathbf{p}_e^{i+1*}$ lies in the second quadrant; \textbf{Case 5}: $\mathbf{p}_e^{i+1*}$ lies in the third quadrant; and \textbf{Case 6}: $\mathbf{p}_e^{i+1*}$ lies in the fourth quadrant.

\begin{figure}[t]
  \includegraphics[width=.7\columnwidth]{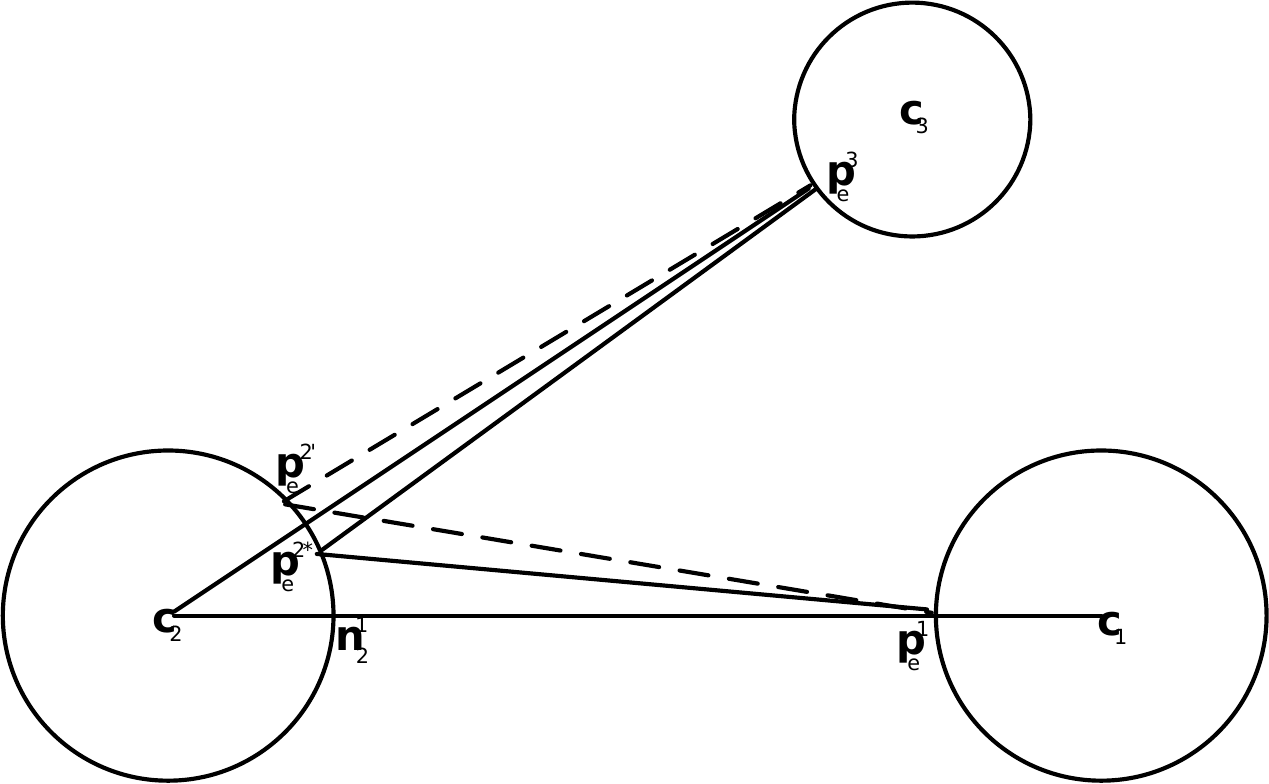}
  \caption{ Illustration for \textit{Case 3}. Clearly, the dotted line path is longer than the solid line path.}
  \label{fig:test1_5}
  \vspace{-6mm}
\end{figure}%
\begin{figure}[t]
  \includegraphics[width=0.9\columnwidth]{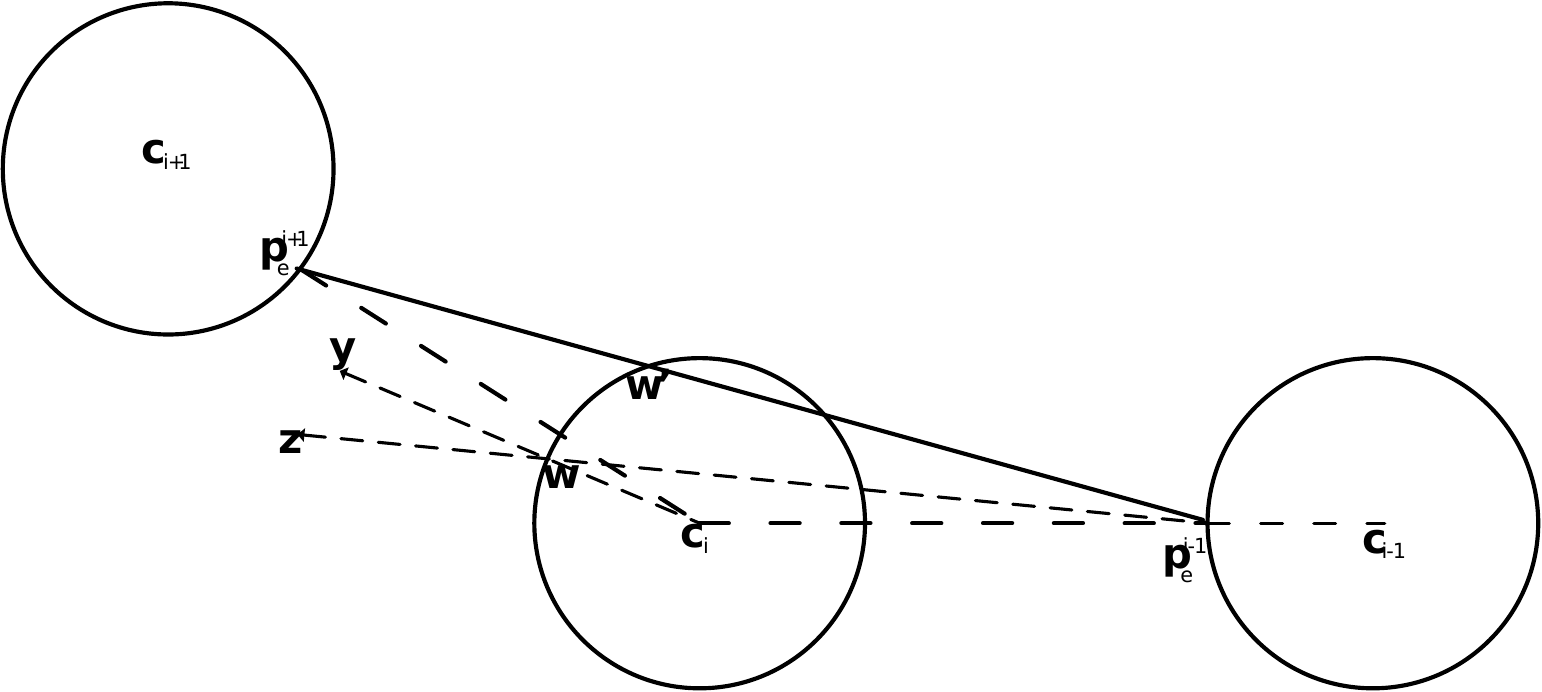}
  \caption{Illustration for \textit{Case 4}. In this figure, we define $\angle \mathbf{c}_{i-1}\mathbf{p}_e^{i-1}\mathbf{p}_e^{i+1}=\theta, \angle \mathbf{c}_{i-1}\mathbf{p}_e^{i-1}\mathbf{w}=\theta', \angle \mathbf{c}_{i-1}\mathbf{c}_{i}\mathbf{p}_e^{i+1}=\phi, \angle \mathbf{c}_{i-1}\mathbf{c}_{i}\mathbf{w}=\phi'$.}
  \label{fig:test2_5}
  \vspace{-6mm}
\end{figure}
Case 3 is depicted in Fig. 13. First, we provide the proof for Case 3 and then for Case 4. The proofs for Case 5 and Case 6 follows similar line of reasoning and hence omitted for brevity.

\textit{Case 3}: Denote the angle between the x-axis and $\mathbf{n}_i^{i+1}$ as $\theta^{i+1}$. Suppose the optimal point, denoted by $\mathbf{x}^*$, lies on the arc between the angles $\theta^{i+1}$ and $2\theta^{i+1}$. Then, based on the symmetry, we can always find another point, $\hat{\mathbf{x}}$, on the arc between angles $0$ and $\theta^{i+1}$ such that $\|\mathbf{x}^*-\mathbf{p}_e^{i+1}\|=\|\hat{\mathbf{x}}-\mathbf{p}_e^{i+1}\|$. Next, we consider the distance between $\hat{\mathbf{x}}$ and $\mathbf{p}_e^{i-1}$. Note that we have
\begin{equation}
    \|\hat{\mathbf{x}}-\mathbf{p}_e^{i-1}\|\!=\!\sqrt{\|\mathbf{c}_i\!-\!\mathbf{p}_e^{i-1}\|^2\!+\!r_i^2\!-\!2\left(r_i\|\mathbf{c}_i\!-\!\mathbf{p}_e^{i-1}\|\right)\cos(\phi)},
\end{equation}
where $0\leq \phi\leq \theta^{i+1}$ is the angle of the point $\hat{\mathbf{x}}$ from the x-axis. 

Taking the first derivative of the right-hand side of (5) with respect to $\phi$ we get
\begin{align}
    \scalemath{.8}{\frac{d \|\hat{\mathbf{x}}\!-\!\mathbf{p}_e^{i-1}\|}{d\phi}}&=\scalemath{.8}{\frac{\|\mathbf{c}_i-\mathbf{p}_e^{i-1}\|r_i\sin(\phi)}{\sqrt{\|\mathbf{c}_i\!-\!\mathbf{p}_e^{i-1}\|^2\!+\!r_i^2\!-\!2r_i\|\mathbf{c}_i\!-\!\mathbf{p}_e^{i-1}\|\cos(\phi)}}}\nonumber \\&\scalemath{.8}{=\frac{\|\mathbf{c}_i-\mathbf{p}_e^{i-1}\|r_i\sin(\phi)}{\|\hat{\mathbf{x}}-\mathbf{p}_e^{i-1}\|}.} 
\end{align}

Clearly, we have $\frac{d \|\hat{\mathbf{x}}-\mathbf{p}_e^{i-1}\|}{d\phi}\geq 0,~ \forall~ 0\leq \phi\leq \pi$. Therefore, the distance between $\mathbf{p}_e^{i-1}$ and $\hat{\mathbf{x}}$ is always smaller than the distance between $\mathbf{p}_e^{i-1}$ and $\mathbf{x}^*$ since the angle corresponding to $\mathbf{x}^*$ is greater than $\phi$. Hence, we have
\begin{align}
    \|\hat{\mathbf{x}}-\mathbf{p}_e^{i-1}\|\leq \|\mathbf{x}^*-\mathbf{p}_e^{i-1}\| & \Rightarrow \|\hat{\mathbf{x}}-\mathbf{p}_e^{i-1}\|+\|\hat{\mathbf{x}}-\mathbf{p}_e^{i+1}\|\nonumber \\ &\leq \|\mathbf{x}^*-\mathbf{p}_e^{i-1}\|+\|\mathbf{x}^*-\mathbf{p}_e^{i+1}\|, 
\end{align}
where the second inequality is the result of choosing appropriate $\hat{\mathbf{x}}$ such that $\|\mathbf{x}^*-\mathbf{p}_e^{i+1}\|=\|\hat{\mathbf{x}}-\mathbf{p}_e^{i+1}\|$. This contradicts the assumption that $\mathbf{x}^*$ is the optimal point. Hence, we conclude that there is no point on the arc between angles $\theta^{i+1}$ and $2\theta^{i+1}$ can be optimal. A similar line of reasoning can be used to show that the optimal point cannot lie on the arc between angles $2\theta^{i+1}$ and $2\pi$. This completes the proof for Case 3.

\textit{Case 4}: For Case 4, two further possibilities are denoted by C4-1 and C4-2. Specifically, C4-1 represents the scenario where $\mathbf{p}_e^{i-1}$ and $\mathbf{p}_e^{i+1}$ can be joined by the line that intersects the circular trajectory twice. This possibility is illustrated in Fig. 14. In this scenario, the optimal connecting point on the $i$-th circular trajectory should lie on the line that joins $\mathbf{p}_e^{i-1}$ and $\mathbf{p}_e^{i+1}$ since any other point will lead to a higher total distance according to the triangle inequality. Next, we observe that this point can only lie on the arc between points $\mathbf{n}_i^{i-1}$ and $\mathbf{n}_i^{i+1}$. To illustrate this observation, consider Fig. 14. Without loss of generality, assume that the optimal leaving point on the $i$-th circular trajectory is on the arc between points $\mathbf{n}_i^{i+1}$ and $[-r_i, 0]$. In Fig. 14, this point is denoted by $\mathbf{w}$. Furthermore, we define the angle from the x-axis of point $\mathbf{n}_i^{i+1}$ from $\mathbf{c}_i=[0,0]$ as $\phi$. Clearly, we have $\phi<\theta<\theta'$. We note that any ray emanating from point $\mathbf{c}_i$ and meeting the line segment $\mathbf{wz}$ will have an angle greater than $\phi$. This means the point $\mathbf{p}_e^{i+1}$ cannot lie on the line segment $\mathbf{wz}$. This contradicts the assumption that any straight line passing through the arc between $\mathbf{n}_i^{i+1}$ and $[-r_i,0]$ can join the points $\mathbf{p}_e^{i-1}$ and $\mathbf{p}_e^{i+1}$. Therefore, points $\mathbf{p}_e^{i-1}$ and $\mathbf{p}_e^{i+1}$ cannot be joined with a straight line segment that passes through the arc that joins points $\mathbf{n}_i^{i+1}$ and $[-r_i, 0]$. Hence, the lemma is proved for sub-case C4-1. On the other hand, C4-2 deals with the possibility where $\mathbf{p}_e^{i-1}$ and $\mathbf{p}_e^{i+1}$ can be joined by the line segments that intersect the circular trajectory only once. For C4-2, using the same line of reasoning used for Case 3, we can show that the optimal point always lies on the arc joining the points $\mathbf{n}_i^{i-1}$ and $\mathbf{n}_i^{i+1}$.

\textit{Case 5 and Case 6:} Using the same reasoning as used for Case 3 and Case 4, we can also proof for Case 5 and Case 6.
\end{proof}
\vspace{-10mm}
\subsection{Proof of Lemma 5:}
\begin{proof}
The proof first shows that the problem of finding the optimal point can be formulated as a convex optimization problem. Then, the proof uses the fact that any local optimal for a convex optimization problem is also the global optimal.

Note that the distance between the connecting point and $\mathbf{p}_e^{i-1}, \mathbf{p}_e^{i+1}$ can be written as
\begin{equation}
    \scalemath{.9}{\|\mathbf{p}_e^{i-1}-\mathbf{p}_e^{i}\|=\sqrt{(\mathbf{p}_e^{i-1}(x)-r_i\cos(\phi))^2+r_i^2\sin^2(\phi)},}
\end{equation}
\begin{equation}
    \scalemath{.8}{\|\mathbf{p}_e^{i+1}\!-\!\mathbf{p}_e^{i}\|\!=\!\sqrt{(\mathbf{p}_e^{i\!+\!1}(x)\!-\!r_i\cos(\phi))^2\!+\!(\mathbf{p}_e^{\!i\!+1}(y)\!-\!r_i\sin(\phi))^2}.}
\end{equation}

It is clear that $l_{i}(\phi)\triangleq\|\mathbf{p}_e^{i-1}-\mathbf{p}_e^{i}\|$, is an increasing function of $\phi$ while $l_{i+1}(\phi)\triangleq \|\mathbf{p}_e^{i+1}-\mathbf{p}_e^{i}\|$ is a decreasing function of $\phi$. Hence, both $l_{i}(\phi)$ and $l_{i+1}(\phi)$ are quasi-convex functions. The optimization problem corresponding to the minimization of $l_{i}(\phi)+l_{i+1}(\phi)$ can be written as
\vspace{-2mm}
\begin{mini}|s|
{\phi^{j+1}}{l_i(\phi^{j+1})+l_{i+1}(\phi^{j+1})}{}{}
\addConstraint{0\leq \phi^{j+1}\leq \theta^{i+1},}
\end{mini}
which, by introducing new variables $M,N$, can be equivalently written as 
\vspace{-5mm}
\begin{mini}|s|
{\phi^{j+1}, M,N}{M+N}{}{}
\addConstraint{M\geq l_i(\phi^{j+1}),}
\addConstraint{N\geq l_{i+1}(\phi^{j+1}),}
\addConstraint{0\leq \phi^{j+1} \leq \theta^{i+1}.}
\end{mini}
Problem (30) can be converted to a convex optimization problem. Therefore, the Bisection search can be utilized to obtain the optimal value of $\phi$. This completes the proof.
\end{proof}
\vspace{-4mm}
\subsection{Proof of Lemma 6:}
\begin{proof}
First, we note that the time needed to traverse $C_i$ is a monotonically decreasing function of $r_i \in [\frac{R_i}{2}, r_{i}^{opt}]$. Hence, by choosing $r_i=r_i^{opt}$, the time needed to traverse $C_i$ is minimized. However, this choice may increase the value of $l_i(r_i)+l_{i+1}(r_i)$. Therefore, the optimal value of $r_i$ in the $j+1$-th iteration can be obtained by solving the following problem
\vspace{-3mm}
\begin{maxi}|s|
{r_i^{j+1}}{r_i^{j+1}}{}{}
\addConstraint{l_i(r_i^{j+1})+l_{i+1}(r_i^{j+1})\leq \zeta^{j+1}}
\addConstraint{r_i^j\leq r_i^{j+1}\leq r_i^{opt},}
\end{maxi}
where $\zeta^{j+1}$ is the optimal value of $l_i(r_i)+l_{i+1}(r_i)$ obtained by performing the optimization over $\phi$ in the $j+1$-th iteration, and $r_i^{j} \in [\frac{R_i}{2}, r_{i}^{opt}]$ is the optimal value of $r_i$ obtained in $j$-th iteration. 

Note that the objective function and the second constraint of (31) are linear functions. In order to establish the convexity/non-convexity of $l_i(r_i)+l_{i+1}(r_i)$, we focus on the distance between a random point in space and a point on the circumference of a circle. We show that the distance between any random point and a point on the circumference of a circle is a convex function of the radius of that circle. Denote a random point by $\mathbf{z}$ and the point on the circle of radius $r$ by $[r\cos(\phi), r\sin(\phi)]$. Then, it follows

\begin{equation}
    \scalemath{.85}{f(r)\triangleq\|\mathbf{z}-[r\cos(\phi), r\sin(\phi)]\|=\|\mathbf{z}(x)-r\cos(\phi), \mathbf{z}(y)-r\sin(\phi)\|}. 
\end{equation}

Recall that if $f(r)$ is a convex function, then we must have

\begin{equation}
    f(\lambda r_1+(1-\lambda)r_2)\leq \lambda f(r_1)+(1-\lambda) f(r_2),
\vspace{-0.1 cm}
\end{equation}
where $\lambda \in [0,1]$. For $f(r)$ to be convex, the necessary and sufficient condition is

\begin{equation}
    \scalemath{.9}{\|\mathbf{z}(x)-(\lambda r_1+(1\!-\!\lambda)r_2)\cos(\phi), \mathbf{z}(y)-(\lambda r_1+(1\!-\!\lambda)r_2)\sin(\phi)\|} \nonumber
\end{equation}
\begin{align}
    &\scalemath{.9}{\leq \lambda \|\mathbf{z}(x)-r_1\cos(\phi), \mathbf{z}(y)-r_1\sin(\phi)\|} \nonumber \\ &\scalemath{.9}{+(1-\lambda)\|\mathbf{z}(x)-r_2\cos(\phi), \mathbf{z}(y)-r_2\sin(\phi)\|.} 
\end{align}

Next, by utilizing the triangle inequality and the non-negativity of $\lambda, 1-\lambda$, we show that (34) is always true. Therefore, $f(r)$ is a convex function. This means $l_i(r_i)\geq 0$ and $l_{i+1}(r_i)\geq 0$ are also convex. Since the non-negative sum of convex functions is also convex, we conclude that problem (31) is a convex optimization problem with a single optimization variable $r_i^{j+1}$. Therefore, we can employ Bisection search to obtain its solution. This completes the proof.
\end{proof}
\vspace{-4mm}
\subsection{Proof of Lemma 7:}
\begin{proof}
Assume that the optimal completion time achieved after $j$-th iteration of \textbf{Algorithm 2} is $T_j$ and the corresponding value of optimal radius is denoted by $r_i^{j*}$. Then, we need to show that $T_{j+1}\leq T_j$. Note that in the $j+1$-th iteration of \textbf{Algorithm 2} only $\mathbf{p}_e^i$ is optimized. Therefore, the only part of the total time that is affected by \textbf{Algorithm 2} in each iteration corresponds to the time needed to traverse $l_i+C_i+l_{i+1}$. Then, line 2 of \textbf{Algorithm 2} achieve a global optimal of $\phi^{j+1}$ that minimizes travel distance $l_i(\phi^{j+1})+l_{i+1}(\phi^{j+1})$. Denote the completion time obtained via this optimization by $T_{j+1}^{\phi^{j+1}}$. Clearly, $T_{j+1}^{\phi^{j+1}}\leq T_{j}$ due to the global optimality of $\phi^{j+1}$. The remaining task is to show that $T_{j+1}^{r_i^{j+1}}\leq T_{j+1}^{\phi^{j+1}}$, where $T_{j+1}^{r_i^{j+1}}$ denotes the completion time achieved by optimization of $r_i^{j+1}$. Note that by the construction of the optimization problem (31) through the involvement of constraint $l_i(r_i^{j+1})+l_{i+1}(r_i^{j+1})\leq \zeta^{j+1}$, the travel time needed to cover the distance $l_i+l_{i+1}$ is smaller than that achieved by selecting $r_i=r_i^{j}$. Since the time needed to cover $C_i$ is a monotonic decreasing function of $r_i \in [\frac{R_i}{2}, r_i^{opt}]$, we can show that $T_{j+1}^{r_i^{j+1}}\leq T_{j+1}^{\phi^{j+1}}$. Hence, we have established that $T_{j+1}\triangleq T_{j+1}^{r_i^{j+1}}\leq T_{j+1}^{\phi^{j+1}}\leq T_j$. This completes the proof.
\end{proof}
\vspace{-6mm}
\subsection{Proof of Proposition 2:}
\begin{proof}
Without loss of generality, consider the $((k-1)I+i)$-th iteration. Assume that the total completion time achieved at the end of the previous iteration is denoted by $T_{(k-1)I+i-1}$ and the corresponding optimal connecting endpoints on each circular trajectory are denoted by $[\mathbf{p}_e^{1,(k-1)I+i-1},\cdots,\mathbf{p}_e^{I,(k-1)I+i-1}]$. Then, in the $((k-1)I+i)$-th iteration, the goal is to minimize the time needed to travel from $\mathbf{p}_e^{i-1,(k-1)I+i-1}$ to $\mathbf{p}_e^{i+1,(k-1)I+i-1}$, which is equal to the sum of travel time corresponding to $l_i+l_{i+1}$ and the coverage time corresponding to the coverage area centered at $\mathbf{c}_i$, by optimizing $\mathbf{p}_e^{i}$ while the rest of the endpoints are chosen according to (35) shown at the top of this page. 
\begin{figure*}
\begin{equation}
[\mathbf{p}_e^{1,(k-1)I+i*},\cdots, \mathbf{p}_e^{i-1,(k-1)I+i*},\mathbf{p}_e^{i+1,(k-1)I+i*},\cdots,\mathbf{p}_e^{I,(k-1)I+i*}]\nonumber
\end{equation}
\begin{equation}
=[\mathbf{p}_e^{1,(k-1)I+i-1*},\cdots, \mathbf{p}_e^{i-1,(k-1)I+i-1*},\mathbf{p}_e^{i+1,(k-1)I+i-1*},\cdots,\mathbf{p}_e^{I,(k-1)I+i-1*}].
\end{equation}
\hrule
\end{figure*}

With a selection of connecting endpoints, only the time for travelling from $\mathbf{p}_e^{i-1,(k-1)I+i*}$ to $\mathbf{p}_e^{i+1,(k-1)I+i*}$ is reduced while the rest of the travel times for all the $\mathbf{p}_e^{j,(k-1)I+i*}$ to $\mathbf{p}_e^{j+1,(k-1)I+i*}$ remain equal to their values at the end of $((k-1)I+i-1)$-th iteration. Hence, the total completion time achieved at the end of $((k-1)I+i)$-th iteration is less than or equal to the total completion time achieved at the end of $((k-1)I+i-1)$-th iteration. Thus, the completion time achieved by \textbf{Algorithm 3} is non-increasing after each iteration. Since, the total completion time is lower bounded, \textbf{Algorithm 3} is guaranteed to converge. This completes the proof. 
\end{proof}
\vspace{-2mm}
\bibliographystyle{IEEEtran}
\bibliography{refs}

\end{document}